%% file: a_main.tex
\documentclass[letter,reqno]{a_preamble}
\title{
A Multiplicative Ergodic Theorem for\\
Bistochastic Ergodic Quantum Processes\\
with Applications to Entanglement
}
\date{}

\usepackage{authblk}
\author{Owen Ekblad\orcidlink{0009-0006-0834-0327}\thanks{ekbladow@msu.edu}}
\affil{Michigan State University, Department of Mathematics}

\begin{document}

\pagenumbering{arabic}
\lhead{\thepage}
\maketitle

\begin{abstract}
We prove a multiplicative ergodic theorem for bistochastic completely positive (bcp) linear cocycles acting on finite-dimensional matrix algebras, giving an invariant splitting described explicitly in terms of the multiplicative domains of the underlying bcp maps.
As an application of our theorem, we classify when compositions of random bcp maps are asymptotically entanglement breaking, and use this classification to show that occasionally positive partial transpose bcp maps are asymptotically entanglement breaking.
We conclude by demonstrating a certain class of bcp linear cocycles are almost surely entanglement breaking in finite time.
\end{abstract}

\section{Introduction}
\input{sa_Introduction}
\section{Multiplicative properties}\label{Sec:Mult}
\input{sc_Multiplicative}
%%%

\section{Final Remarks}\label{Sec:Final}
\input{sd_Final}
\appendix 
%

%%%

%
\section{The Grassmannian and the Fell \texorpdfstring{$\sigma$-algebra}{l}}\label{App:Grassmannian}
\input{sdapp_Grass}
\section{Stopping times}\label{App:Meas}
\input{sdapp_Stopping}

\subsection*{Acknowledgments}
The author is grateful to Jeffrey Schenker for useful mathematical discussions related to the content of this paper, and to anonymous reviewers whose comments greatly improved the presentation.
The author completed this work while supported by research assistantships from Professor Jeffrey Schenker with funding from the US National Science Foundation under Grant No. DMS-2153946.

\printbibliography

\end{document}

%% file: sa_Introduction.tex
Consider an open quantum system $\mcS$ described by a finite-dimensional Hilbert space $\scrH = \mbC^d$. 
One simple model describing the dynamics of such a system is given by a quantum channel $\psi$, which is a completely positive and trace-preserving linear map 
\begin{equation}
    \psi:\matrices\to\matrices, 
\end{equation}
where $\matrices$ denotes the set $d\times d$ matrices with complex entries:
the dynamics is defined by the rule that if, at time $t = 0$, the state of $\mcS$ is described by the density matrix $\rho\in\matrices$, then the state of $\mcS$ at time $t = n$ is given by $\psi^n(\rho)$ \cite{Cusumano2022QuantumGuide, Ciccarello2022QuantumInteractions, Davies1976QuantumSystems, Kraus1983StatesTheory}. 
In this work, we consider a stochastic version of this model, where the quantum channel defining our dynamics is sampled from an ergodic stochastic process:
if $\mcS$ is described by the density matrix $\rho$ at time $t = 0$, then the state of $\mcS$ at time $t = n$ is described by 
\begin{equation}
    \phi_{{n-1}; \omega}\circ\cdots\circ\phi_{0; \omega}(\rho), 
\end{equation}
where $\omega\in\Omega$ is an element of a standard probability space $\seq{\Omega, \mcF, \mu}$, and $\Phi := \seq{\phi_n}_{n\in\mbZ}$ is a bi-infinite sequence of dynamically-defined random quantum channels, i.e., for all $n\in\mbZ$, 
\begin{equation}
    \begin{split}
        \phi_n:\Omega&\to\scrQ\\
        \phi_{n; \omega}
            &:=
        \phi_{T^n(\omega)},
    \end{split}
\end{equation}
where $\scrQ$ denotes the set of quantum channels, $T:\Omega\to\Omega$ is an invertible ergodic $\mu$-preserving transformation, and $\phi:\Omega\to\scrQ$ is a prescribed random quantum channel.
To refer to all this data succinctly, we call $\Phi$ the \textit{ergodic quantum process} defined by $\seq{T, \phi}$.
Examples of ergodic quantum processes include the situations where $\Phi$ is an independent identically distributed sequence or a Markovian sequence (see Example \ref{Ex:Markovian}), but also models situations with long-range stochastic correlations in time, like periodic or quasiperiodic regimes (see Example \ref{Ex:Nondeterministic}).
Understanding ergodic quantum processes under various assumptions has been the focus of recent literature, motivated on multiple fronts: 
as we have just described, ergodic quantum processes describe (disordered) open quantum dynamics, but they also describe the theory of matrix product states on quantum spin chains with homogeneously-distributed disorder \cite{Movassagh2021TheoryProcesses, Movassagh2022AnStates}.
In this work, we take as motivation the following simple question. 
\begin{description}
\hypertarget{ENT}{}
    \item[\ENT]  How is entanglement preserved under the random quantum dynamics defined by $\Phi$?
\end{description}
This is a basic question of interest in applications to quantum information theory and quantum communication \cite{Horodecki2003EntanglementChannels}, but has also proved to be a mathematically interesting question in its own right, evidenced by the substantial amount of mathematical literature dedicated to some form of \ENT\, \cite{Szczygielski2024EventuallyDynamics, Lami2015EntanglementbreakingIndices, Christandl2019WhenBreaking, Lami2016Entanglement-savingChannels, Rahaman2018EventuallyMaps, Kennedy2018CompositionMaps, Strmer2008SeparableMaps}.
The literature, however, has primarily been concerned with answering \ENT\, under the implicit assumption of \textit{no disorder}, and, to the author's knowledge, no substantial work has been done on addressing \ENT\, when there is external disorder affecting the system. 
In this work, we seek to remedy this gap in the literature, first providing a theoretical framework for addressing \ENT\, and problems like it in the disordered regime, and then beginning to address \ENT\, for ergodic quantum processes. 
Our main assumption in this work is the following: 
\begin{description}
\hypertarget{BCPa}{}
    \item[\BCPa]  For almost every $\omega\in\Omega$, $\phi_\omega\in\bcp$, where $\bcp$ denotes the set of unital quantum channels.
\end{description}
Here, unitality refers to the fact that $\phi_\omega(\mbI) = \mbI$, where $\mbI\in\matrices$ is the identity matrix. 
We call an element $\psi\in\bcp$ a bistochastic completely positive (\textit{bcp}) map, where bistochasticity refers to $\psi$ being both trace-preserving and unital. 
The assumption of bistochasticity covers many situations of interest in applications, in addition to giving a useful $C^*$-algebraic framework we use to address \ENT.\footnote{We address in more depth the extent to which the bcp assumption is used in Section \ref{Sec:Final}.}
Indeed, the methodology we follow to answer \ENT\, is the same used by others authors under the assumption of no disorder: 
in \cite{Rahaman2018EventuallyMaps}, specifically, the authors address \ENT\, by using a $C^*$-algebraic decomposition result from \cite{Rahaman2017MultiplicativeChannels} that gives a clean description of $\matrices$ in terms of the dynamics defined by $\Phi$. 
We follow the same methodology here:
our first two results (Theorems \ref{Thm:Deterministic} and \ref{Thm:Main theorem}) generalize the aforementioned decomposition result (\cite[Theorem 2.5]{Rahaman2017MultiplicativeChannels}) to the disordered setting, which then gives us a powerful technical tool to address \ENT, which we do in Theorems \ref{Thm:Abelian} and \ref{Thm:PPT}.
As we now describe, however, there are substantial technical adjustments that must be made to accommodate disorder. 
Let us be more precise.
\subsection{Multiplicative ergodic theory}
The primary technical fact we use is that the data $\seq{T, \phi}$ defining $\Phi$ above is a \textit{linear cocycle}, where we recall (from, e.g., \cite[Ch. 2]{Viana2014LecturesExponents}) that a linear cocycle is defined by a pair $\seq{\theta, A}$ where $\theta:\Omega\to\Omega$ is a probability-preserving map and $A:\Omega\to\mbM_D$ is a random matrix, and the cocycle is understood as a map $\seq{\theta, A}:\Omega\times\mbC^D\to\Omega\times\mbC^D$ defined by
\begin{equation}
    \seq{\theta, A}: (\omega, v)\mapsto (\theta(\omega), A_\omega v).
\end{equation}
A fundamental theorem concerning linear cocycles is the classical Multiplicative Ergodic Theorem (\textit{MET}) of Oseledets \cite{Oseledets1968ASystems}.
%\footnote{A notable precursor to this theorem is the 1960 work of Furstenburg and Kesten \cite{Furstenberg1960ProductsMatrices}, which was concerned with the Lyapunov exponents of linear cocycles.} 
%
The refinement of the MET in the case of $\seq{T, \phi}$ forms the technical backbone of this work, and so it is useful to recall the statement of the MET here.
\begin{thmx}[MET \cite{Oseledets1968ASystems}]
\label{Thm:Mult Erg Theorem}
    Let $\seq{\theta, A}$ be a linear cocycle, and let $\matnorm{\cdot}$ and $\|\cdot\|$ be norms making $\mbM_D$ and $\mbC^D$, respectively, into Banach spaces.
    Assume that 
    \begin{equation}\label{Eqn:Integ condition for Oseledets}
        \int_\Omega 
            \log^+\matnorm{A_\omega}\,\dee\mu(\omega)
        <\infty,
    \end{equation}
    where $\log^+$ denotes the function $\max(\log, 0)$.
    Then there exist real numbers $\lambda_1 > \cdots > \lambda_k$ (with $\lambda_k$ equal to $-\infty$) and measurable $\seq{\theta, A}$-invariant subspace-valued random variables
    \begin{equation}
        0 \subseteq V^{\leq\lambda_k}
        \subsetneq \cdots \subsetneq V^{\leq\lambda_1} = \mbC^D
    \end{equation}
    such that for $\mu$-almost every $\omega\in\Omega$, all $j\in\{1, \dots, k\}$, and any $v\in V_{\omega}^{\leq\lambda_j}\setminus V_{\omega}^{\leq\lambda_{j+1}}$, the equation
    \begin{equation}
        \lim_{n\to\infty}
        \frac{1}{n}
        \log 
        \left\|
            A_{\theta^{n-1}(\omega)}\cdots A_\omega v
        \right\|
        =
        \lambda_{j}
    \end{equation}
    is satisfied.
\end{thmx}
For a standard treatment of this theorem, see \texorpdfstring{\cite[Theorem 10.2]{Walters1982AnTheory}}{l}. 
A subspace-valued random variable $V\subseteq \mbC^D$ is called $\seq{\theta, A}$-invariant if 
$
    A_\omega V_\omega \subseteq V_{\theta(\omega)} 
$
holds for almost every $\omega\in\Omega$. 
In the above theorem, the numbers $\lambda_1 > \cdots > \lambda_k$ are called the \textit{Lyapunov exponents} associated to the linear cocycle $\seq{\theta, A}$, and, collectively, the set $\set{\lambda_j}$ is called the \textit{Lyapunov spectrum}.
We call $\lambda_1$ the \textit{top Lyapunov exponent} and we call $\lambda_j$ the \textit{$j$th Lyapunov exponent}.
The subspace-valued random variables $V^{\leq\lambda_k}\subsetneq \cdots \subsetneq V^{\leq\lambda_1}$ are called the \textit{Lyapunov subspaces} associated to the linear cocycle $\seq{\theta, A}$.\footnote{Two notes about the version of the MET we state here: (1) usually, the MET is stated for matrices and vector spaces over $\mbR$, but it is easy to see that when $\mbC$ is isometrically identified with $\mbR^2$, the statement of the MET we gave here remains unchanged; (2) we have taken the convention that $\lambda_k = -\infty$ is always a Lyapunov exponent and that $V^{\leq-\infty}$ may be equal to $\set{0}$, which is nonstandard but streamlines our presentation.}
Let $\bcp$ denote the set of bcp maps. 
To refine the MET in the case of $\seq{T, \phi}$, we make use of the \textit{multiplicative domain}:
for any $\psi\in\bcp$, the multiplicative domain $\mcM_\psi$ is the unital $C^*$-algebra
\begin{equation}
    \mcM_\psi 
        :=
    \set{a\in\matrices\,\,:\,\, \psi(ab) = \psi(a)\psi(b)\text{ and }\psi(ba) = \psi(b)\psi(a)\text{ for all }b\in\matrices}.
\end{equation}
It has been shown by several authors that $\mcM_\psi$ is a useful object for understanding the information-theoretic properties of $\psi$ \cite{Choi2009TheCorrection, Johnston2011GeneralizedCorrection}, and in general $\mcM_\psi$ is a useful tool in the general $C^*$-algebraic analysis of completely positive maps \cite{Choi1974AC-Algebras, Strmer2008SeparableMaps}.
Most relevant to our present investigation, however, is the usage of the multiplicative domain in the classification of dynamical behavior of $\Phi$, where, in the deterministic setting, the following theorem of Rahaman is particularly pertinent.
Let $\innerhs{a}{b} =\tr{a^*b}$ denote the Hilbert-Schmidt inner product, and, given a linear map $\psi:\matrices\to\matrices$, let $\psi^*$ denote the adjoint of $\psi$ with respect to $\innerhs{\cdot}{\cdot}$. 
\begin{thmx}[\texorpdfstring{\cite[Theorem 2.5]{Rahaman2017MultiplicativeChannels}}{l}]
\label{Thm:Rahaman}
    Let $\psi\in\bcp$. 
    Then there is $N\in\mbN$ such that the following hold. 
    \begin{enumerate}[label=(\alph*)]
        \item The equality
        \begin{equation}
           \mcM_{\psi^N}^\perp
            =
            \set{a\in\matrices\,\,:\,\,
            \lim_{n\to\infty}
            \|\psi^n(a)\|
            =
            0
            }
        \end{equation}
        holds, where $\mcM_{\psi^N}^\perp = \set{a\in\matrices\,\,:\,\,\innerhs{b}{a} = 0\text{ for all }b\in\mcM_{\psi^N}}$ and $\|\cdot\|$ is any norm making $\matrices$ into a Banach space. 

        \item  $\psi\seq{\mcM_{\psi^N}} = \mcM_{\psi^N}$, and, moreover, the $\psi\vert_{\mcM_{\psi^N}}:\mcM_{\psi^N}\to\mcM_{\psi^N}$ defines a $*$-isomorphism of $C^*$-algebras with inverse $\psi^*\vert_{\mcM_{\psi^N}}$.

        \item $\mcM_{\psi^N}$ is the $C^*$-algebra generated by the set 
        \begin{equation}\label{Eqn:Rahaman:Intro}
            \set{a\in\matrices\,\,:\,\,
            \psi(a) = \lambda a \text{ for some $\lambda\in\mbT$}
            },
        \end{equation}
        where $\mbT = \set{\lambda\in\mbC\,\,:\,\,|\lambda| = 1}$. 
    \end{enumerate}
\end{thmx}
By viewing $\psi$ as the trivial linear cocycle $\seq{\operatorname{Id}, \psi}$ (where $\operatorname{Id}$ is the identity map on a one-point probability space), we may interpret Theorem \ref{Thm:Rahaman} as a version of the MET that says $\mcM_{\psi^N}^\perp = V^{\leq\lambda_2}$ and that $\seq{\operatorname{Id}, \psi}$ acts unitarily on $\mcM_{\psi^N}$, in addition to giving an explicit description of $\mcM_{\psi^N}$ in terms of the eigenmatrices of $\psi$. 
The smallest $N$ in the above theorem is called the \textit{multiplicative index} of $\psi$.
Our first technical result generalizes this concept of multiplicative index to the disordered setting, and establishes one of its most basic properties.
For the sake of simplifying our presentation in this introduction, we take the following assumption which we shall later drop in the main body. 
\begin{description}
\hypertarget{Indep}{}
    \item[\Indep]  The $\sigma$-algebras $\mcF^{<0} := \sigma(\phi_n\,\,:\,\,n < 0)$ and $\mcF^{\geq 0} := \sigma(\phi_n\,\,:\,\, n\geq 0)$ are such that $\mu[A] \in\set{0, 1}$ for all $A\in\mcF^{<0}\cap\mcF^{\geq 0}$. 
\end{description}
This assumption is satisfied, for example, in the situation that $\Phi$ is an i.i.d. sequence.
Now, some notation and terminology: 
we say that a random variable $\tau:\Omega\to\mbN\cup\set{\infty}$ is a $\Phi$-stopping time if $\set{\tau = n}\in\sigma(\phi_0, \dots, \phi_{n-1})$ for all $n\in\mbN$, where $\sigma(\phi_0, \dots, \phi_{n-1})$ is the $\sigma$-algebra generated by $\phi_0, \dots, \phi_{n-1}$.
Given such a time with $\tau<\infty$ almost surely, we define 
\begin{equation}
    \begin{split}
        \Phi^{(\tau)}:\Omega&\to\bcp\\
        \Phi^{(\tau)}_\omega &:=
        \phi_{\tau(\omega)-1; \omega}\circ\cdots\circ\phi_{0; \omega}.
    \end{split}
\end{equation}
A special case of this notation is when $\tau = n$ almost surely, so $\Phi^{(n)} = \phi_{n-1}\circ\cdots\circ\phi_0$.
We let $\mcM_{\Phi^{(\tau)}}$ denote the subspace-valued random variable defined by $\mcM_{\Phi^{(\tau)}; \omega} := \mcM_{\Phi_\omega^{(\tau)}}$.\footnote{For measurability concerns about $\Phi^{(\tau)}$ and $\mcM_{\Phi^{(\tau)}}$, see Appendices \ref{App:Meas} and \ref{App:Grassmannian}, respectively}
%
% Our first theorem concerns the generalization of the multiplicative index to the disordered situation. 
%
%
\begin{restatable}[Stabilization of Multiplicative Domain]{thm}{Deterministic}\label{Thm:Deterministic}
    Assume \BCPa\,  and \Indep. 
    Then there is a $\Phi$-stopping time $\tau$ and a deterministic $C^*$-algebra $\mcA_\Phi\subseteq\matrices$  such that $\mcM_{\Phi^{(\tau)}} = \mcA_\Phi$ almost surely. 
\end{restatable}
In keeping with \cite{Rahaman2017MultiplicativeChannels}, we call the $\Phi$-stopping time $\tau$ in the above theorem the \textit{multiplicative index} of $\Phi$, and we call the deterministic algebra $\mcA_\Phi$ the \textit{stabilized multiplicative domain} of $\Phi$.
\textit{A priori}, there is no reason that the subspace-valued random variable $\mcM_{\Phi^{(n)}}$ need be deterministic for any $n$, and that this is in fact true is the primary content of Theorem \ref{Thm:Deterministic}.
Our next result---our main technical theorem---describes how $\mcA_\Phi$ encodes dynamical information of $\Phi$, giving a full extension of Theorem \ref{Thm:Rahaman} to the disordered situation. 
\begin{restatable}[MET for Bistochastic Ergodic Quantum Processes]{thm}{Oseledets}\label{Thm:Main theorem}
Assume \BCPa\, and \Indep. 
Let $\tau$ and $\mcA_\Phi$ denote the multiplicative index and stabilized multiplicative domain of $\Phi$, respectively.
    \begin{enumerate}[label = (\alph*)]
        \item For almost every $\omega\in\Omega$, the equality
        \begin{equation}
            \mcA_\Phi^\perp
            =
           V^{\leq\lambda_2}_\omega
        \end{equation}
        holds, where $\lambda_2<0$ is the second Lyapunov exponent associated to the linear cocycle $\seq{T, \phi}$ and $V^{\leq\lambda_2}$ is the corresponding Lyapunov subspace. 
        In particular, $V^{\leq \lambda_2}$ is almost surely constant.

        \item For almost every $\omega\in\Omega$, $\phi_\omega\!\seq{\mcA_\Phi} = \mcA_\Phi$, and $\phi_\omega\vert_{\mcA_\Phi}:\mcA_\Phi\to \mcA_\Phi$
    defines a $*$-isomorphism of $C^*$-algebras with inverse $\phi_\omega^*\vert_{\mcA_\Phi}$. 

    \item For almost every $\omega\in\Omega$, $\mcA_\Phi$ is the $C^*$-algebra generated by the set 
    \begin{equation}
        \set{
        a\in\matrices\,\,:\,\,
        \Phi^{(\tau)}_\omega(a) = \lambda a \text{ for some }\lambda\in\mbT
        },
    \end{equation}
    where $\mbT = \set{\lambda\in\mbC\,\,:\,\,|\lambda| = 1}$. 
    \end{enumerate}
\end{restatable}
When compared with the MET, Theorems \ref{Thm:Deterministic} and \ref{Thm:Main theorem} together say that, under \BCPa\, and \Indep\, $V^{\leq\lambda_2}$ is deterministic and that $V^{\leq\lambda_1} = V^{\leq\lambda_2}\oplus\mcA_\Phi$, with $\mcA_\Phi$ invariant under $\seq{T, \phi}$. 
This gives a useful description of the dynamics of $\seq{T, \phi}$. 
Moreover, the above theorems vastly extend Theorem \ref{Thm:Rahaman}, first showing that, under \Indep, the stabilized multiplicative domain is necessarily deterministic, and then demonstrating that, for almost every $\omega\in\Omega$, if $a\in\matrices$ is such that 
\begin{equation}
    \lim_n\|\phi_{n; \omega}\circ\cdots\circ\phi_{0; \omega}(a)\|
    =
    0,
\end{equation}
then the rate of convergence is necessarily exponential, which, although obvious in the deterministic case\footnote{For quantum channels, the spectrum of $\psi$ is contained in the unit circle $\mbT\subset\mbC$, and therefore by the singular value decomposition of $\psi$ as a linear map, we see that if $\psi^n(a)$ converges to zero it does so at an exponential rate dictated by the second largest singular value of $\psi$.}, is not \textit{a priori} clear in the random case.
Our Theorem \ref{Thm:More general main theorem} gives a more general version of Theorem \ref{Thm:Main theorem} without assuming \Indep, where the main difference is that the limiting object, $\mcA_\Phi$, need not be deterministic anymore, which can be clearly seen in Example \ref{Ex:Nondeterministic} below.
These new results serve as a useful tool for addressing the entanglement question, as we now discuss. 
\subsection{Applications to entanglement}
Recall that a linear map $\psi:\matrices\to\matrices$ is called \textit{entanglement breaking} if, for any $k\in\mbN$ and any positive semidefinite matrix $\rho\in\matrices\otimes \mbM_k$, the positive semidefinite matrix $\psi\otimes\operatorname{Id}_{\mbM_k}(\rho)$ is separable, meaning that $\psi\otimes\operatorname{Id}_{\mbM_k}(\rho)$ belongs to the convex cone of $\matrices\otimes \mbM_k$ generated by the elements $P\otimes Q$ where $P\in\matrices$ and $Q\in\mbM_k$ are positive semidefinite \cite{Horodecki2003EntanglementChannels}.
Heuristically, maps which are entanglement breaking are unable to take advantage of quantum entanglement to accomplish communication tasks. 
The precise version of \ENT\, we seek to answer, therefore, is the following qualitative question: 
if we let $\eb$ denote the set of entanglement breaking bcp maps, under what circumstances does $\Phi$ approach $\eb$?
That is, we seek to classify the behavior
\begin{equation}\label{Eqn:Asym_ent_breaking}
    \lim_{n\to\infty}d_{\operatorname{HS}}\seq{\Phi^{(n)}_\omega, \eb}
    =
    0,
\end{equation}
where $d_{\operatorname{HS}}\seq{\cdot, \eb}:\bcp\to[0, \infty)$ denotes the distance function $d_{\operatorname{HS}}\seq{\psi, \eb} = \inf_{\varphi\in\eb}\hsnorm{\psi - \varphi}$ and $\hsnorm{\cdot} = \sqrt{\innerhs{\cdot}{\cdot}}$ is the norm induced by the Hilbert-Schmidt inner product.
We say that $\Phi_\omega$ is \textit{asymptotically entanglement breaking (a.e.b.)} at $\omega\in\Omega$ if $\omega$ satisfies (\ref{Eqn:Asym_ent_breaking}). 
As has been noted by other authors, entanglement breaking properties are intimately related to the multiplicative domain (see, e.g., \cite{Strmer2008SeparableMaps} and \cite{Rahaman2018EventuallyMaps}). 
Specifically, in the deterministic case, it was shown in \cite[Theorem 6.1]{Rahaman2018EventuallyMaps} that a given bcp map $\psi\in\bcp$ satisfies $\lim_{n\to\infty}d_{\operatorname{HS}}\seq{\psi^n, \eb} = 0$
if and only if the stabilized multiplicative domain of $\psi$ is abelian. 
Our first result towards \ENT\, generalizes this fact. 
Let $X_{\operatorname{aeb}}$ be the set of $\omega\in\Omega$ such that $\Phi_\omega$ is a.e.b. at $\omega$.\footnote{The set $X_{\operatorname{aeb}}$ is clearly measurable, since $\eb$ is a closed (hence Borel) subset of the set of linear maps $\matrices\to\matrices$.}
\begin{restatable}[Classification of a.e.b.]{thm}{Abelian}\label{Thm:Abelian}
    Assume \BCPa\, and \Indep. 
    Let $\mcA_\Phi$ be the stabilized multiplicative domain of $\Phi$.
    Then $\mu[X_{\operatorname{aeb}}]\in\set{0, 1}$.
    Moreover, $\mu[X_{\operatorname{aeb}}] = 1$ if and only if $\mcA_\Phi$ is abelian. 
\end{restatable}
In accordance with the above theorem, if $\Phi$ is such that $\mu[X_{\operatorname{aeb}}] = 1$, we may say that $\Phi$ is asymptotically entanglement breaking (\textit{a.e.b.}) without specifying $\omega\in\Omega$. 
The criterion for a.e.b. above gives an interesting corollary about positive partial transpose (\textit{PPT}) maps, which, with the goal of resolving of Christandl's $\mathrm{PPT}^2$ conjecture \cite{Christandl2012PPTConjecture}, have received a great deal of attention.
Recall that a linear map $\psi:\matrices\to\matrices$ is called \textit{PPT} if completely positive and completely copositive, meaning that both $\psi$ and $\operatorname{(\cdot)^T}\circ\psi$ are completely positive, where $(\cdot)^T:\matrices\to\matrices$ denotes the transpose with respect to the computational basis of $\mbC^d$. 
\begin{conj*}[\texorpdfstring{$\mathrm{PPT}^2$ Conjecture}{l} \cite{Christandl2012PPTConjecture}]
    Whenever $\psi:\matrices\to\matrices$ is PPT, $\psi^2$ is entanglement breaking. 
\end{conj*}
The $\mathrm{PPT}^2$ conjecture has been resolved in many special cases (see \cite{Collins2018TheStates, Chen2019Positive-partial-transpose3, Christandl2019WhenBreaking, Singh2022TheMaps, Nechita2024RandomChannels, gulati2025entanglementcyclicsigninvariant}), but remains unresolved in full generality. 
Our main contribution in this direction is that if products of bcp maps contained PPT maps with only positive (and possibly very small) probability, then this product is a.e.b.:
\begin{restatable}[Occasionally PPT implies a.e.b.]{thm}{PPT}\label{Thm:PPT}
    Assume \BCPa.   
    If $\mu[\text{$\phi$ is PPT}\,]>0$, then $\Phi$ is asymptotically entanglement breaking. 
\end{restatable}
Notice, in particular, that even if the $\mathrm{PPT}^2$ conjecture holds, it would \textit{not} imply the above theorem.
Indeed, the framework of ergodic quantum processes supports the possibility that, for almost every $\omega\in\Omega$ and all $n\in\mbZ$, at most one of $\phi_{n; \omega}$ and $\phi_{n+1; \omega}$ is PPT, even if $\mu[\text{$\phi$ is PPT}\,]>0$, as we show in Example \ref{Ex:Markovian}.  
A natural next step in this analysis is to ask not just about asymptotic entanglement breaking, but about \textit{eventual} entanglement breaking, which is more useful in practice. 
For $\omega\in\Omega$, we say that $\Phi$ is eventually entanglement breaking (\textit{e.e.b.}) at $\omega$ if there exists $N\in\mbN$ such that $\Phi^{(N)}_\omega\in\eb$, and, as above, we let $X_{\operatorname{eeb}}$ denote the set of all $\omega\in\Omega$ such that $\Phi$ is e.e.b. at $\omega$. 
We define 
\begin{equation}
    \begin{split}
        \iota:\Omega&\to\mbN\cup\set{\infty}\\
        \iota(\omega)
            &:=
        \inf\set{n\in\mbN\,\,:\,\,\Phi^{(n)}_\omega\in\eb},
    \end{split}
\end{equation}
where we define $\inf\emptyset = \infty$, and we call $\iota$ the \textit{index of separability} of $\Phi$. 
Note that $X_{\operatorname{eeb}} = \set{\iota < \infty}$.
Our results about the index of separability are much more preliminary than those on the multiplicative index, but nevertheless our technical framework allows us to obtain some results in a relatively straightforward manner. 
First of our results is the following 0-1 law for $\iota$: 
\begin{restatable}{thm}{Separability}\label{Thm:Index_of_serparbility}
    Assume \BCPa.    
    Then the index of separability $\iota$ defines a $\Phi$-stopping time such that $\mu[\iota < \infty]\in\set{0, 1}$. 
\end{restatable}
As in the case of a.e.b., we say that $\Phi$ is e.e.b. if $\iota<\infty$ almost surely.
The next natural question, then, is under what conditions is $\Phi$ e.e.b.?
We leave this question open in full generality, but we are able to prove a particular class of bistochastic ergodic quantum processes are e.e.b.
\begin{restatable}{thm}{EEB}\label{Thm:EEB}
    Assume \BCPa\, and \Indep. 
    Let $\mcA_\Phi$ denote the stabilized multiplicative domain of $\Phi$. 
    If $\mcA_\Phi = \mbC\mbI$, then $\Phi$ is e.e.b.
\end{restatable}
\subsection{Relation to other works}\label{Subsec:Relation}
As we mentioned at the start of the introduction, this work falls under the theoretical umbrella of so-called ergodic quantum processes. 
The systematic study of fully general ergodic quantum processes was laid out by Movassagh and Schenker in the joint physics and math papers \cite{Movassagh2021TheoryProcesses, Movassagh2022AnStates}, and has since been receiving increased attention by various authors, as mentioned in the body above.
Here, Movassagh and Schenker studied linear cocycles $\seq{T, \varphi}$ where $T$ was an invertible ergodic measure-preserving transformation (as above), but the only assumption was that $\seq{T, \varphi}$ is \textit{eventually strictly positive}, i.e., for almost every $\omega\in\Omega$, there exists $N\in\mbN$ such that $\varphi_{T^N(\omega)}\circ\cdots\circ\varphi_\omega(\rho)$ is invertible for all density matrices $\rho$. 
In particular, there is no requirement that $\varphi$ be trace-preserving or unital almost surely. 
In Proposition \ref{Prop:ESP} below, we realize that the assumption of eventual strict positivity is equivalent to $\mcA_\Phi = \mbC\mbI$ in the bistochastic case, and therefore Theorem \ref{Thm:EEB} is complementary to the work of Movassagh and Schenker.
Movassagh and Schenker were originally motivated by questions arising in the theory of matrix product states, but, as we have seen here, the framework they introduced in \cite{Movassagh2021TheoryProcesses, Movassagh2022AnStates} has a much wider range of applicability than just this. 
See \cite{Movassagh2021TheoryProcesses, Movassagh2022AnStates, Nelson2024ErgodicAlgebras, Ekblad2024ReducibilityInteractions, Souissi2025ErgodicProcesses} for a selection of recent works on ergodic quantum processes.
Movassagh and Schenker, however, were not the first to study compositions of random quantum channels: 
this has been studied for many years under the name of \textit{repeated interactions}, which are motivated more by general open quantum dynamics questions. 
These have been considered by many authors in both the physics and math literature; see \cite{Bruneau2010InfiniteDynamics, Nechita2012RandomStates, Bruneau2014RepeatedSystems, Bougron2022MarkovianSystems} for a non-exhaustive list of such works. 
On the random dynamical systems side, ergodic quantum processes and random repeated interactions may all be understood as a class of positivity-preserving linear cocycles on a finite-dimensional noncommutative $C^*$-algebra. 
The dynamical properties of linear cocycles that preserve positive definite matrices is nascent, and, to this author's knowledge, has to date mostly been studied under the guise of ergodic quantum processes and random repeated interactions, but the study of positivity-preserving linear cocycles on finite-dimensional \textit{commutative} $C^*$-algebras has been much better studied. 
These results are concerned with properties of products of random matrices with non-negative (and often strictly positive) entries, which may be understood as random compositions of positive maps acting on a finite-dimensional \textit{commutative} $C^*$-algebra $\mbC^d$; 
the interested reader may consult the non-exhaustive list of references
\cite{Evstigneev1974PositiveSystems, Arnold1994EvolutionaryMatrices, Kifer1996Perron-FrobeniusEnvironments, Hennion1997LimitMatrices, Pollicott2010MaximalProducts}
for further reading related to this topic. 
There has also been some work done on positive linear cocycles on ordered Banach spaces \cite{Mierczynski2013PrincipalTheory, Mierczynski2013PrincipalSystems, Mierczynski2016PrincipalSystems}.
Although we didn't explicitly mention this above, our Theorems \ref{Thm:Abelian} and \ref{Thm:PPT} extend \cite[Theorem 6.1]{Rahaman2018EventuallyMaps} and \cite[Theorem 4.4]{Rahaman2018EventuallyMaps}, respectively, to compositions of random bcp maps.
In general, the theory of entanglement breaking maps has received much attention. 
In \cite{Horodecki2003EntanglementChannels}, Horodecki, Shor, and Ruskai established the fundamental theory of entanglement breaking maps on matrix algebras, which shows that entanglement breaking maps cannot be used for any communication task that utilizes quantum entanglement in a fundamental way \cite{Bauml2015LimitationsRepeaters, Christandl2017PrivateSecrecy}. 
The question of asymptotic and eventual entanglement breaking has been addressed in \cite{Lami2015EntanglementbreakingIndices, Lami2016Entanglement-savingChannels, Rahaman2018EventuallyMaps, Kennedy2018CompositionMaps, Hanson2020EventuallyTimes, hanson2020entropiccontinuitybounds, Szczygielski2024EventuallyDynamics} and elsewhere. 
Many of these works take advantage of the $C^*$-algebraic framework of positive linear maps on general $C^*$-algebras, and most relevant to our work is the theory developed by St{\o}rmer in the works \cite{Stormer1963PositiveAlgebras, Strmer2007MultiplicativeMap, Strmer2008SeparableMaps}.
\subsection{Organization}
The remainder of this paper is dedicated to giving the full technical details needed to prove all our results. 
As mentioned above, we shall prove more general results in the case that \Indep\, does not necessarily hold, but we shall always assume \BCPa. 
We begin Section \ref{Sec:Mult} by giving the full definitions of all the objects whose meanings were left implicit in the above introduction.
Then, we prove the main technical results required to prove Theorems \ref{Thm:Deterministic} and \ref{Thm:Main theorem}, and the more general Theorem \ref{Thm:More general main theorem}. 
Next, we change gears and turn to the question of entanglement breaking, where we once more recall all relevant definitions and facts required to formally prove our main theorems here (the more general versions of Theorems \ref{Thm:Abelian}, \ref{Thm:PPT}, \ref{Thm:Index_of_serparbility}, and \ref{Thm:EEB} without assuming \Indep).
We conclude with some final remarks in Section \ref{Sec:Final}, where we discuss the extent to which the bistochastic assumption is used and how one might extend these results to the non-bistochastic regime, in addition to stating a generalization to the i.i.d. case of a result due to Kuperberg that didn't fit well within the main body presentation. 
In our appendices \ref{App:Grassmannian} and \ref{App:Meas}, we include the details required to show that pertinent subspace-valued random variables and stopping times are measurably defined, in addition to recalling some basic technical facts about the Grassmannian as a metric space.

%% file: sc_Multiplicative.tex
\subsection{\texorpdfstring{$C^*$-algebraic preliminaries}{l}}
We begin by setting notation and establishing preliminaries that will be used throughout the rest of this paper. 
We write $\matrices$ to denote the space of $d\times d$ matrices with entries in $\mbC$. 
For $a\in\matrices$, we write $a^*$ to denote the conjugate transpose of $a$, we write $\operatorname{Spec}(a)$ to denote the set of eigenvalues of $a$, and we write $\tr{a}$ to denote the trace of $a$. 
We call $a\in\matrices$ \textit{positive semidefinite} if $a = a^*$ and $\operatorname{Spec}(a)\subset[0, \infty)$, and we write $a\geq 0$ to denote this. 
For $a, b\in\matrices$, we write $a\geq b$ or $b\leq a$ to denote that $a - b \geq 0$.
We let $\mbI$ denote the identity matrix in $\matrices$. 
The space $\matrices$ is a Banach space with various norms---all of which are topologically equivalent by merit of the finite-dimensionality of $\matrices$ as a $\mbC$-vector space---but there are only two norms on $\matrices$ to which we shall refer in this work: the operator norm and the norm induced by the Hilbert-Schmidt inner product. 
The operator norm $\infnorm{\cdot}$ is the norm defined by 
\begin{equation}
    \infnorm{a}
    =
    \sup_{v\in\mbC^d\setminus\set{0}}
    \cfrac{\|av\|}{\|v\|},
\end{equation}
where for a vector $v\in\mbC^d$, $\|v\|$ denotes the usual norm induced by the standard Euclidean inner product on $\mbC^d$. 
It is a standard fact that, when equipped with the operator norm $\infnorm{\cdot}$, $\matrices$ has the structure of a $C^*$-algebra. 
The Hilbert-Schmidt inner product $\innerhs{\cdot}{\cdot}$ is the inner product on $\matrices$ defined by 
\begin{equation}
    \innerhs{a}{b}
        =
    \tr{a^*b} \quad a, b\in\matrices. 
\end{equation}
It is a standard fact that $\innerhsNoArg$ makes $\matrices$ into a Hilbert space.
For a subset $S\subset\matrices$, we write $S^\perp$ to denote the orthogonal complement of $S$, i.e., the linear space defined by 
\begin{equation}
    S^\perp 
        =
    \set{a\in\matrices\,\,:\,\,
    \innerhs{a}{s} = 0\text{ for all $s\in S$}}.
\end{equation}  
We write $\hsnorm{\cdot}$ to denote the norm induced by $\innerhsNoArg$.
Unless otherwise stated, all maps $\matrices\to\matrices$ mentioned in this work are assumed to be $\mbC$-linear, and we write $\superops$ to denote the set of all such maps. 
Given $\psi\in\superops$, we write $\infnorm{\psi}$ (resp. $\hsnorm{\psi}$) to denote the operator norm on $\psi$ induced by $\infnorm{\cdot}$ (resp. $\hsnorm{\cdot}$), i.e., 
\begin{equation}
    \|\psi\|_{\#}
        =
    \sup_{a\in\matrices\setminus\set{0}}
    \cfrac{\|\psi(a)\|_{\#}}{\|a\|_{\#}}
\end{equation}
where $\# \in\{ \infty, \operatorname{HS}\}$. 
As before, these norms generate the same topology, because the space $\superops$ is a finite-dimensional $\mbC$-vector space. 
Specifically, there are constants $C_1, C_2>0$ such that 
\begin{equation}
    C_1\hsnorm{\psi}\leq \infnorm{\psi}\leq C_2\hsnorm{\psi},
\end{equation}
a fact which we use freely in the following. 
Whenever we refer to a measurability structure on $\superops$, we shall always refer to the Borel $\sigma$-algebra defined by these operator norms. 
In this work, we are concerned with a particular subset of $\superops$: specifically, the set of bistochastic completely positive maps. 
Recall that a map $\psi\in\superops$ is called \textit{positive} $\psi(a)\geq 0$ whenever $a\geq 0$, and we call $\psi$ \textit{completely positive} if for any $k\in\mbN$, the map 
\begin{equation}
    \begin{split}
        \psi\otimes\operatorname{Id_{\mbM_k}}:
        \matrices\otimes\mbM_k &\to \matrices\otimes\mbM_k\\
        a\otimes b&\mapsto \psi(a)\otimes b
    \end{split}
\end{equation}
is positive (to define positive semidefiniteness for an element $c\in\matrices\otimes\mbM_k$, one identifies $\matrices\otimes\mbM_k$ with $\mbM_{dk}$).
It is straightforward to check that if $\psi$ is positive, then $\psi(a^*) = \psi(a)^*$ for any $a\in\matrices$, a fact which we implicitly use throughout.
A map $\psi\in\superops$ is called \textit{trace preserving} if for any $a\in\matrices$, the identity
\begin{equation}
    \tr{\psi(a)}
    =
    \tr{a}
\end{equation}
holds, and $\psi$ is called \textit{unital} if $\psi(\mbI) = \mbI$. 
We call $\psi$ \textit{bistochastic} if $\psi$ is both unital and trace preserving.
In this work, we consider the collection of bistochastic completely positive (\textit{bcp}) maps frequently enough that some notation is warranted: we write $\bcp$ to denote the set
\begin{equation}
    \bcp
        =
    \set{\psi\in\superops\,\,:\,\,\text{$\psi$ is bcp}}.
\end{equation}
It is straightforward to check that $\bcp$ is a compact subset of $\superops$, a fact which we use freely in the following. 
For $\psi\in\bcp$, we write $\mcM_\psi$ to denote the \textit{multiplicative domain} of $\psi$, 
\begin{equation}
    \mcM_\psi 
        =
    \set{a\in\matrices\,\,:\,\,
    \psi(ab) = \psi(a)\psi(b)\text{ and }\psi(ba) = \psi(b)\psi(a)
    \text{ for all $b\in\matrices$}}.
\end{equation}
It is well-known that $\mcM_\psi$ is a $C^*$-algebra on which $\psi$ acts as a unital $*$-homomorphism. 
The following theorem of Choi gives an alternative description of $\mcM_\psi$. 
\begin{thmx}[\texorpdfstring{\cite{Choi1974AC-Algebras}}{l}]\label{Thm:Choi}
    For a unital completely positive map $\psi:\matrices\to\matrices$, the equality 
    \begin{equation}
        \mcM_\psi 
            =
        \set{a\in\matrices\,\,:\,\,
        \psi(a^*a)
        =
        \psi(a)^*\psi(a)\text{ and }
        \psi(aa^*) = \psi(a)\psi(a)^*}
    \end{equation}
    holds.
\end{thmx}
In \cite{Choi1974AC-Algebras}, Choi proves a more general result regarding the multiplicative domain of unital linear maps on general $C^*$-algebras satisfying positivity conditions weaker than complete positivity, but the version of Theorem \ref{Thm:Choi} is all that is required for the purposes of this work. 
A central property of unital completely positive maps used to conclude results of this sort is the \textit{Schwarz inequality}: 
given a unital completely positive map $\psi\in\superops$, $\psi$ satisfies 
\begin{equation}
    \psi(a^*a)
    \geq 
    \psi(a)^*\psi(a)
\end{equation}
for all $a\in\matrices$, a fact that is proved in \cite{Choi1974AC-Algebras}. 
We make frequent reference to this inequality, which in particular holds for all $\psi\in\bcp$. 
This has the following useful consequence. 
\begin{lemma}\label{Lem:Mult dom of bcp map}
    For $\psi\in\bcp$, 
    \begin{equation}
        \mcM_\psi 
            =
        \set{a\in\matrices 
            \,\,:\,\,
        \|a\|_{\operatorname{HS}}
        =
        \|\psi(a)\|_{\operatorname{HS}}
        }
    \end{equation}
    and 
     \begin{equation}
        \mcM_{\psi}^\perp
        \setminus\{0\}
            \subset
        \set{a\in\matrices\,\,:\,\,
            \|a\|_{\operatorname{HS}}
            >
            \|\psi(a)\|_{\operatorname{HS}}
        }.
    \end{equation}
    In particular, $\hsnorm{\psi}\leq 1$ for any $\psi\in\bcp$. 
\end{lemma}
\begin{proof}
    Assume $a\in\mcM_\psi$. 
    Then from Theorem \ref{Thm:Choi}, $\psi(a^*a) = \psi(a)^*\psi(a)$, so since $\psi$ is trace preserving, we have that $\|a\|_{\operatorname{HS}} = \|\psi(a)\|_{\operatorname{HS}}$. 
    Conversely, suppose $\|a\|_{\operatorname{HS}} = \|\psi(a)\|_{\operatorname{HS}}$.
    By the Schwarz inequality, we know that $\psi(a^*a) - \psi(a)^*\psi(a)\geq 0,$ so since
    \begin{equation}
        \tr{\psi(a^*a) - \psi(a)^*\psi(a)}
        =
        \|a\|_{\operatorname{HS}}^2 - \|\psi(a)\|_{\operatorname{HS}}^2
        =
        0,
    \end{equation}
    we may conclude that $\psi(a^*a) - \psi(a)^*\psi(a) = 0$.
    Because $\|b\|_{\operatorname{HS}} = \|b^*\|_{\operatorname{HS}}$ for any $b\in\matrices$, we may argue the same way to conclude that $\psi(aa^*) = \psi(a)\psi(a)^*$, so by Theorem \ref{Thm:Choi}, we have that $a\in\mcM_{\psi}$. 
    To see the claim about $\mcM_{\psi}^\perp
        \setminus\{0\}$, note that by the Schwarz inequality, we have that 
    \begin{equation}
        \|a\|_{\operatorname{HS}}^2 = \tr{a^*a}
        =
        \tr{\psi(a^*a)}
        \geq 
        \tr{\psi(a)^*\psi(a)}
        =
        \|\psi(a)\|_{\operatorname{HS}}
    \end{equation}
    for any $a\in\matrices$. 
    So, since $\mcM_{\psi}^\perp\cap \mcM_{\psi} = \set{0}$, $a\in \mcM_{\psi}^\perp\setminus\{0\}$ implies $\|a\|_{\operatorname{HS}} > \|\psi(a)\|_{\operatorname{HS}}$.
    It is immediate from the above that $\hsnorm{\psi}\leq 1$ for any $\psi\in\bcp$, so the proof is concluded. 
\end{proof}
\begin{remark}
    It is clear that $\psi\in\bcp$ satisfies $\hsnorm{\psi}\leq 1$ directly from the Schwarz inequality, since the trace-preservation of $\psi$ yields $\tr{\psi(a)^*\psi(a)}\leq \tr{\psi(a^*a)} = \tr{a^*a}$, but we include the statement of this fact (and its alternative proof) in the above lemma for succinctness. 
\end{remark}
Another basic fact we use in this work is the following, which is proved in \cite{Rahaman2017MultiplicativeChannels}, but for completeness's sake we include a proof here also. 
\begin{lemma}\label{Lem:Comp of mult doms}
    For $\psi_1, \psi_2\in\bcp$, $\mcM_{\psi_2\circ\psi_1} = \set{a\in\mcM_{\psi_1}\,\,:\,\, \psi_1(a)\in\mcM_{\psi_2}}$. 
    Thus, $\mcM_{\psi_2\circ\psi_1}\subseteq \mcM_{\psi_1}$. 
\end{lemma}
\begin{proof}
    Note that the composition of bcp maps is bcp. 
    In general, for any $a\in\matrices$, by the Schwarz inequality and the trace preservation of $\psi_1$ and $\psi_2$, we have the chain of inequalities
    \begin{equation}\label{Eqn:Comp_of_mult_doms_eqn_1}
        \tr{\psi_2\circ\psi_1(a)^*\psi_2\circ\psi_1(a)}
        \leq 
        \tr{\psi_1(a)^*\psi_1(a)}
        \leq 
        \tr{a^*a},
    \end{equation}
    i.e., $\hsnorm{\psi_2\circ\psi_1(a)}\leq \hsnorm{\psi_1(a)}\leq \hsnorm{a}$.
    So, by Lemma \ref{Lem:Mult dom of bcp map}, because $a\in\mcM_{\psi_2\circ\psi_1}$ if and only if $\hsnorm{a} = \hsnorm{\psi_2\circ\psi_1(a)}$, the inequalities in (\ref{Eqn:Comp_of_mult_doms_eqn_1}) become equalities, and another application of Lemma \ref{Lem:Mult dom of bcp map} concludes our proof. 
\end{proof}
As a result of this lemma, for any sequence $\Psi = \seq{\psi_n}_{n\geq 0}\subset\bcp$, we have the decreasing chain of $C^*$-algebras
\begin{equation}
    \mcM_{\psi_0}\supseteq 
    \mcM_{\psi_1\circ\psi_0}\supseteq\cdots\supseteq 
    \mcM_{\psi_n\circ\cdots\circ\psi_0}\supseteq\cdots,
\end{equation}
which, because $\matrices$ is a finite-dimensional $\mbC$-vector space, shows that there is $N\in\mbN$ for which
\begin{equation}\label{Eqn:Stab_Mult_Dom_defn}
    \bigcap_{n\in\mbN}
    \mcM_{\Psi^{(n)}}
        =
    \mcM_{\Psi^{(N)}},
\end{equation}
where $\Psi^{(n)} = \psi_{n-1}\circ\cdots\circ\psi_0$ for all $n\in\mbN$. 
\begin{definition}[Stabilized multiplicative domain and multiplicative index]
    For a sequence $\Psi = \seq{\psi_n}_{n\geq 0}$ with $\psi_n\in\bcp$ for all $n$, we call the intersection
    \begin{equation}\label{Eqn:Stabl_mult_dom}
        \bigcap_{n\in\mbN}
    \mcM_{\Psi^{(n)}}
    \end{equation}
    the stabilized multiplicative domain of $\Psi$ and we write $\StabMultDom{\Psi}$ to denote it. 
    If $\Psi = \seq{\psi_n}_{n\in\mbZ}$ is a bi-infinite sequence, we write $\StabMultDom{\Psi}$ to denote the same intersection as in (\ref{Eqn:Stabl_mult_dom}).
    In the special case that $\Psi_0 = \seq{\psi^n}_{n\in\mbN}$ for a given $\psi\in\bcp$, we write $\StabMultDom{\psi}$ to denote $\StabMultDom{\Psi_0}$.
    In general, we call the minimal number $N\in\mbN$ such that the equality (\ref{Eqn:Stab_Mult_Dom_defn}) holds the \textit{multiplicative index} of $\Psi$. 
\end{definition}
Before we move to the random situation described in the introduction, let us prove a basic lemma in this generality.
Given a bi-infinite sequence $\Psi = \seq{\psi_n}_{n\in\mbZ}$, we let $S(\Psi) = \seq{\psi_{n+1}}_{n\in\mbZ}$ be the shifted sequence. 
\begin{lemma}\label{Lem:inhomo_mult_cores}
    Let $\Psi = \seq{\psi_n}_{n\in\mbZ}$ be a bi-infinite sequence of bcp maps. 
    Then $\psi_0\seq{\StabMultDom{\Psi}} \subseteq \StabMultDom{S(\Psi)}$.
    Moreover, $\psi_0:\StabMultDom{\Psi}\to \StabMultDom{S(\Psi)}$ defines a unital $*$-homomorphism that is isometric with respect to both $\infnorm{\cdot}$ and $\hsnorm{\cdot}$.  
\end{lemma}
\begin{proof}
    It is clear from the definition of $\StabMultDom{\Psi}$ that $\psi_0\vert_{\StabMultDom{\Psi}}$ acts $*$-homomorphically, so we just need to prove the other claims. 
    We begin by showing that $\psi_0(\StabMultDom{\Psi})\subseteq \StabMultDom{S(\Psi)}$.
    Towards this end, let $a\in\StabMultDom{\Psi}$. 
    Then for all $n\in\mbN$, $a\in\mcM_{\Psi^{(n+1)}}\cap\mcM_{\psi}$. 
    Noting that $\Psi^{(n+1)} = \seq{S(\Psi)}^{(n)}\circ\psi_0$, we conclude from Lemma \ref{Lem:Mult dom of bcp map} that 
    \begin{equation}
        \hsnorm{\seq{S(\Psi)}^{(n)}\seq{\psi_0(a)}}
        =
        \hsnorm{\Psi^{(n+1)}(a)}
        =
        \hsnorm{a}
        =
        \hsnorm{\psi_0(a)}
    \end{equation}
    for all $a$ and $n\in\mbN$. 
    So, again by Lemma \ref{Lem:Mult dom of bcp map}, we have $\phi_0(a)\in \StabMultDom{S(\Psi)}$. 
    For the part about $\psi_0$ acting isometrically, note that $\hsnorm{\phi_0(a)} = \hsnorm{a}$ for any $a\in\StabMultDom{\Psi}\subseteq\mcM_{\phi_0}$ by Lemma \ref{Lem:Mult dom of bcp map}---hence $\phi_0\vert_{\StabMultDom{\Psi}}$ is isometric with respect to $\hsnorm{\cdot}$---and consequently $\phi_0\vert_{\StabMultDom{\Psi}}$ is injective.
    So, since $\psi_0\vert_{\StabMultDom{\Psi}}$ is a $*$-homomorphism between $C^*$-algebras (see, e.g., \cite[Theorem 2.1.7]{Murphy2007C-AlgebrasTheory}), $\infnorm{\phi_0(a)}\leq \infnorm{a}$ for all $a\in\StabMultDom{\Psi}$. 
    However, again because $\phi_0\vert_{\StabMultDom{\Psi}}$ is a $*$-homomorphism, $\phi_0\seq{\StabMultDom{\Psi}}$ is a $C^*$-algebra (see \cite[Theorem 3.1.6]{Murphy2007C-AlgebrasTheory}), so because $\phi_0^{-1}:\phi_0\seq{\StabMultDom{\Psi}}\to \StabMultDom{\Psi}$ is a $*$-homomorphism, so again by \cite[Theorem 2.1.7]{Murphy2007C-AlgebrasTheory}, we conclude that $\infnorm{\phi_0(a)}\geq \infnorm{a}$ for all $a\in\StabMultDom{\Psi}$, hence $\phi_0\vert_{\StabMultDom{\Psi}}$ is a $\infnorm{\cdot}$-isometry, as claimed. 
\end{proof}
\subsection{Random multiplicative domains and the multiplicative index}
In this work, we consider $\Psi = \seq{\psi_n}_{n\in\mbZ}$ where the $\psi_n$ satisfy global regularity conditions described by an ergodic dynamical system, and we study the associated multiplicative domains and indices. 
Let us make this precise. 
Let $\seq{\Omega, \mcF, \mu}$ be a standard probability space, and let $T:\Omega\to\Omega$ be an invertible ergodic measure preserving transformation, which means that $\mu \seq{T^{-1}(E)} = \mu(E)$ for all $E\in\mcF$ and that $\mu$ satisfies the ergodic hypothesis for $T$, i.e., 
\begin{equation}
    \text{for all $E\in\mcF$ with $T^{-1}(E) = E$ it is necessary that $\mu[E]\in\set{0, 1}$.}
\end{equation}
Note that the invertibility of $T$ implies that $T^{-1}$ is also an ergodic measure preserving transformation. 
Let us state some common useful reformulations of ergodicity that we will make free use of in the following. 
\begin{prop}[Equivalent formulations of ergodicity]\label{Prop:Basics_of_ergodicity}
    Let $\seq{\Omega, \mcF, \mu}$ be a probability space and let $T:\Omega\to\Omega$ be an invertible measure preserving transformation. 
    The following are equivalent. 
    \begin{enumerate}[label = (\alph*)]
        \item $T$ is ergodic.

        \item $T^{-1}$ is ergodic.

        \item For any $E\in\mcF$ with $T^{-1}(E)\subset E$, $\mu[E]\in\set{0, 1}$.

        \item For any measurable function $f:\Omega\to\mbR$, the almost sure equality $f\circ T = f$ implies that $f$ is constant. 
    \end{enumerate}
\end{prop}
\begin{proof}
    See \cite[Ch. 4]{Viana2015FoundationsTheory} or \cite[Ch. 1]{Walters1982AnTheory}.
\end{proof}
Let $\phi:\Omega\to\bcp$ define a $\mcF$-measurable bcp map. 
For any $n\in\mbN$, we let $\Phin:\Omega\to\bops{\matrices}$ denote the map defined by 
\begin{equation}
\Phin_\omega:=
\phi_{T^{n-1}(\omega)}\circ\cdots\circ\phi_{\omega}
    % \begin{cases}
    %     \phi_{T^{n-1}(\omega)}\circ\cdots\circ\phi_{\omega} &\text{$n\geq 1$}\\
    %     \operatorname{Id}_{\matrices} &\text{$n = 0$}\\
    %     \phi_{T^{-1}(\omega)}\circ\cdots\circ\phi_{T^{-n-1}(\omega)} &\text{$n \leq -1$}.
    % \end{cases}
\end{equation}
Because $\phi$ is bcp almost surely, for all $n\in\mbZ$ the map $\Phin$ is bcp almost surely, since the composition of bcp maps is bcp.
We let $\Phi_\omega$ denote the sequence $\seq{\phi_{n; \omega}}_{n\in\mbN}$, and we let $\StabMultDom{\Phi; \omega}$ denote the stabilized multiplicative domain associated to $\Phi_\omega$. 
We have thus defined an $\omega$-dependent subalgebra of $\matrices$, and it is necessary to understand the sense in which this algebra depends measurably on $\omega$. 
Letting $\mbG\seq{\matrices}$ denote the Grassmannian of $\matrices$, i.e., 
\begin{equation}
    \mbG\seq{\matrices} 
        =
    \set{W\subseteq\matrices\,\,:\,\, 
    \text{$W$ is a $\mbC$-linear subspace of $\matrices$}},
\end{equation}
then equipping $\mbG\seq{\matrices}$ with the Grassmannian metric 
\begin{equation}
    d_{\mbG\seq{\matrices}}\seq{W, U}
    =
     \max 
    \Bigg( 
        \sup_{w\in W_1}\inf_{u\in U} \hsnorm{w - u},
        \sup_{u\in U_1}\inf_{w\in W}\hsnorm{w - u}
    \Bigg),
\end{equation}
where for a subset $S\subseteq\matrices$, $S_1 = \set{s\in S\,\,:\,\, \hsnorm{s} = 1}$, we endow $\mbG\seq{\matrices}$ with a Borel $\sigma$-algebra. 
\begin{definition}[Subspace-valued random variable]
    We say a map to or from $\mbG\seq{\matrices}$ is measurable if it is measurable with respect to the Borel $\sigma$-algebra on $\mbG\seq{\matrices}$ defined by $d_{\mbG\seq{\matrices}}$.
    We call a measurable map $\mcV:\Omega\to\mbG\seq{\matrices}$ a subspace-valued random variable. 
\end{definition}
In Appendix \ref{App:Grassmannian}, we prove the following technical fact. 
\begin{lemma}\label{Lem:Meas of the Ms} 
    The following maps are measurable. 
    \begin{enumerate}[label = (\alph*)]
        \item The map $\mcM:\bcp\to\mbG\seq{\matrices}$ defined by $\psi\mapsto\mcM_\psi$. 

        \item The map $\StabMultDom{\phantom{\cdot}}:\bcp\to\mbG\seq{\matrices}$ defined by $\psi\mapsto\StabMultDom{\psi}$.

        \item The map $\StabMultDom{\Phi}:\Omega\to\mbG\seq{\matrices}$ defined by $\omega\mapsto\StabMultDom{\Phi; \omega}$. 

        \item The map $\Omega\to\mbN$ defined by $\omega\mapsto \operatorname{dim}\StabMultDom{\Phi; \omega}$. 
    \end{enumerate}
\end{lemma}
\begin{proof}
    See Corollary \ref{Cor:App:All Ms are Grassmannian measurable} for proofs of (a), (b), and (c), and note that (d) follows immediately from (c), since the dimension map $\mbG\seq{\matrices}\to\mbN$ is clearly measurable (and in fact continuous) upon identification of $\mbG\seq{\matrices}$ with the set of orthogonal projections in $\scrL$.
\end{proof}
With Lemma \ref{Lem:Meas of the Ms} in hand, we may now begin to prove our main results concerning $\StabMultDom{\Phi}$. 
First, let us notice that, since $\Phi_{T(\omega)} = S(\Phi_\omega)$, Lemma \ref{Lem:inhomo_mult_cores} implies the following for $\Phi$. 
\begin{lemma}\label{Lem:phi omega is C* hom of local cores}
    For almost every $\omega\in\Omega$,  $\phi_\omega\!\seq{\StabMultDom{\Phi; \omega}} \subseteq \StabMultDom{\Phi; T(\omega)}$.
    Moreover, $\phi_\omega:\StabMultDom{\Phi; \omega}\to \StabMultDom{\Phi; T(\omega)}$ defines a unital $*$-homomorphism that is isometric with respect to both $\infnorm{\cdot}$ and $\hsnorm{\cdot}$.  
\end{lemma}
As a corollary, we are able to conclude that the dimension of $\LocalMultCore{\omega}$ is actually deterministic: 
\begin{cor}\label{Cor:Orbit dimension constant on orbits}
    There is $\delta\in\mbN$ such that $\operatorname{dim}\LocalMultCore{\omega} = \delta$ for almost every $\omega\in\Omega$. 
\end{cor}
\begin{proof}
    Let $k\in\mbN$ and let $E_k$ be the set 
    \begin{equation}
        E_k
            =
        \set{\omega\in\Omega\,\,:\,\, \operatorname{dim}\LocalMultCore{\omega} = k}.
    \end{equation}
    By Lemma \ref{Lem:Meas of the Ms}, $E_k$ is measurable. 
    Thus, by Poincar\'e recurrence, for almost every $\omega\in E_k$, there is $N\in\mbN$ such that $T^N(\omega)\in E_k$. 
    So, by the injectivity implied by Lemma \ref{Lem:phi omega is C* hom of local cores}, we conclude that 
    \begin{equation}
       k= \operatorname{dim}\LocalMultCore{\omega}
        \leq 
        \operatorname{dim}\LocalMultCore{T(\omega)}
        \leq 
        \cdots
        \leq
        \operatorname{dim}\LocalMultCore{T^N(\omega)}
        =
        k,
    \end{equation}
    which shows that $\operatorname{dim}\LocalMultCore{\omega} = \operatorname{dim}\LocalMultCore{T(\omega)}$ for almost every $\omega\in E_k$. 
    That is, $T(E_k)\subseteq E_k$. 
    Therefore, by Proposition \ref{Prop:Basics_of_ergodicity}, we see that $\mu(E_\delta) = 1$, which concludes the proof. 
\end{proof}
Henceforth, we shall assume that all statements made about $\omega\in\Omega$ refer to those $\omega$ such that 
\begin{equation}
    \operatorname{dim}\StabMultDom{\Phi; T^k(\omega)} = \delta
\end{equation}
for all $k\in\mbZ$; this is a measurable probability one set by the above corollary as $T$ and $T^{-1}$ are measure preserving.
\begin{cor}\label{Cor:Unitary}
    For almost every $\omega\in\Omega$, the map $\phi_\omega:\LocalMultCore{\omega}\to\LocalMultCore{T(\omega)}$ defines a unital $*$-isomorphism that is isometric with respect to both $\infnorm{\cdot}$ and $\hsnorm{\cdot}$. 
    Moreover, $\phi_\omega^*\!\seq{\LocalMultCore{T(\omega)}} = \LocalMultCore{\omega}$, and $\phi_\omega^*\vert_{\LocalMultCore{T(\omega)}}$ is the inverse of $\phi_\omega\vert_{\LocalMultCore{\omega}}$.
\end{cor}
\begin{proof}
    That $\phi_\omega$ is a unital $*$-isomorphism with the asserted isometric properties follows immediately from  Lemma \ref{Lem:phi omega is C* hom of local cores} and Corollary \ref{Cor:Orbit dimension constant on orbits}.
    Let $\psi_\omega:\StabMultDom{\Phi; T(\omega)}\to\StabMultDom{\Phi; \omega}$ denote the inverse map of $\phi_\omega\vert_{\StabMultDom{\Phi; \omega}}$. 
    We show that $\psi_\omega = \phi_\omega^*\vert_{\StabMultDom{\Phi; T(\omega)}}$, which will conclude the proof. 
    So, let $a\in\StabMultDom{\Phi; T(\omega)}$. 
    We show that $\phi_\omega^*(a) = \psi_\omega(a)$.
    To see this, let $c\in\matrices$ and compute 
     \begin{align*}
     \innerhs{c}{\phi_\omega^*(a)}
        =
        \innerhs{\phi_\omega(c)}{a}
        &= 
        \innerhs{\phi_\omega(c)}{\phi_\omega\seq{\psi_\omega(a)}}
        \\
        &= 
        \tr{
        \phi_\omega(c)^*\phi_\omega\seq{\psi_\omega(a)}
        }\\
        &= 
        \tr{\phi_\omega(c^*\psi_\omega(a))} &&\text{since $\psi_\omega(a)\in\mcM_{\phi_\omega}$}\\
        &= 
        \innerhs{c}{\psi_\omega(a)},
    \end{align*}
    where the last equality holds since $\phi_\omega$ is trace-preserving. 
    Thus, because $c\in\matrices$ was arbitrary, it holds that $\phi_\omega^*(a) = \psi_\omega(a)$, as claimed.
    Since $a\in\StabMultDom{\Phi; T(\omega)}$ was arbitrary, we conclude that $\psi_\omega = \phi_\omega^*\vert_{\StabMultDom{\Phi; T(\omega)}}$, as desired.  
\end{proof}
Now, by the Artinian property of $\matrices$, we know that for all $\omega\in\Omega$, the decreasing chain of algebras 
\begin{equation}
    \mcM_{\Phi^{(1)}_\omega}\supseteq \mcM_{\Phi^{(2)}_\omega}\supseteq\cdots\supseteq\mcM_{\Phin_\omega}
    \supseteq\cdots 
\end{equation}
must eventually stabilize; 
let $\tau:\Omega\to\mbN$ denote the function 
\begin{equation}
    \tau(\omega) 
        := 
    \min
    \set{n\in\mbN\,\,:\,\,\mcM_{\Phin_\omega} = \mcM_{\Phi^{(n+m)}_\omega}\text{ for all $m\in\mbN$}}.
\end{equation}
In keeping with \cite{Rahaman2017MultiplicativeChannels}, we call $\tau$ the \textit{multiplicative index} of $\Phi$. 
The key property of $\tau$ is that 
\begin{equation}
    \LocalMultCore{\omega}
    =
    \mcM_{\Phi^{(\tau)}_\omega}
\end{equation}
for all $\omega\in\Omega$. 
Furthermore, if we let $\mcT = \seq{\mcT_n}_{n\in\mbN}$ be the filtration of $\sigma$-subalgebras of $\mcF$ defined by 
\begin{equation}
    \mcT_n
        =
    \sigma\Big(
        \phi_m\,\,:\,\, m = 0, \dots, n-1
    \Big),
\end{equation}
then actually $\tau$ defines a $\mcT$-stopping time, which we prove in Lemma \ref{Lem:App:Meas of tau}.
At this juncture, we are prepared to see how, under the assumption \Indep\, given in the introduction, we can prove Theorem \ref{Thm:Deterministic}. 
Let us recall the theorem statement. 
\Deterministic*
\begin{proof}
    We already know that the multiplicative index $\tau$ is a $\Phi$-stopping time, and that $\mcM_{\Phi^{(\tau)}} = \StabMultDom{\Phi}$ almost surely, so it suffices to show that $\StabMultDom{\Phi}$ is almost surely constant.
    To do this, we begin by recalling that $\mcF^{<0}$ was the $\sigma$-algebra generated by the random bcp maps $\phi_n$ for all $n < 0$, and that $\mcF^{\geq 0}$ was defined similarly. 
    Note that by the definition of $\StabMultDom{\Phi}$, we have that $\StabMultDom{\Phi}$ is $\mcF^{\geq 0}$-measurable by its very expression as 
    \begin{equation}
        \StabMultDom{\Phi; \omega}
        =
        \bigcap_{n\in\mbN}
        \mcM_{\phi_{n-1; \omega}\circ\cdots\circ\phi_{0; \omega}}.
    \end{equation}
    Therefore, by the assumption \Indep, to show that $\StabMultDom{\Phi}$ is almost surely constant, it suffices to show that there is an $\mcF^{<0}$-measurable subspace-valued random variable $\mcM'$ such that $\mcM' = \StabMultDom{\Phi}$ almost surely. 
    So, the main fact we need to establish is the existence of such $\mcM'\in\mcF^{<0}$.
    To do this, we make an approximation argument, which takes a bit of work to set up. 
    As a start, note that from Corollary \ref{Cor:Orbit dimension constant on orbits} that there is $\delta\in\mbN$ such that $\operatorname{dim}\StabMultDom{\Phi} =\delta$ almost surely. 
    We can use this to our advantage by recalling the canonical identification of $\mbG_\delta\seq{\matrices}$ with $\scrP_\delta$, where 
    \begin{equation}
        \scrP_\delta 
        =
        \set{P\in\scrL\,\,:\,\, P = P^2 = P^*\text{ and }\operatorname{dim}\operatorname{ran}(P) = \delta},
    \end{equation}
    That is, $\scrP_\delta$ is the set of rank $\delta$ orthogonal projections in the set of linear maps $\seq{\matrices, \innerhs{\cdot}{\cdot}}\to \seq{\matrices, \innerhs{\cdot}{\cdot}}$.
    We denote this identification by the map 
    \begin{equation}
        \begin{split}
            \operatorname{Proj}:\mbG_\delta\seq{\matrices} &\to\scrP_\delta\\
            W&\mapsto \operatorname{Proj}(W),
        \end{split}
    \end{equation}
    where $\operatorname{Proj}(W)\in\scrL$ is the unique orthogonal projection onto $W$. 
    From Lemma \ref{Lem:App:Isometry_Grassmannian} in Appendix \ref{App:Grassmannian}, we know that if we equip $\scrP_\delta$ with the metric $d_\scrP(\cdot, \cdot)$ defined by 
    \begin{equation}
        d_\scrP(P, Q)
        =
        \hsnorm{P - Q}
    \end{equation}
    where $\hsnorm{P}$ is the operator norm of $P\in\scrL$ with respect to $\hsnorm{\cdot}$ on $\matrices$, then we have that $\operatorname{Proj}$ is an isometry between $\seq{\mbG_\delta\seq{\matrices}, d_{\mbG\seq{\matrices}}}$ and $\seq{\scrP_\delta, d_\scrP}$.
    The advantage of this identification is that $\operatorname{Proj}\seq{\StabMultDom{\Phi}}:\Omega\to\scrP_\delta$ may be understood as a $\mbM_{d^2}$-valued random variable, and, therefore, because $\operatorname{Proj}\seq{\StabMultDom{\Phi}}$ is $\mcF^{\geq 0}$-measurable, we have by the martingale convergence theorem \cite[Theorem 4.2.11]{Durrett2019Probability:Examples} that 
    \begin{equation}
        \lim_{n\to\infty}\mbE\left[
           \operatorname{Proj}\seq{\StabMultDom{\Phi}}
           \,\,\vert\,\,
           \mcF^{[0, n]}
        \right]
        =
        \operatorname{Proj}\seq{\StabMultDom{\Phi}}
    \end{equation}
    holds $\mu$-almost surely, where $\mcF^{[0, n]}$ is the $\sigma$-algebra generated by $\set{\phi_0, \dots, \phi_n}$. 
    (Note that all the entries of $\operatorname{Proj}\seq{\StabMultDom{\Phi}}$ are integrable since $\opnorm{\operatorname{Proj}\seq{\StabMultDom{\Phi}}}\leq 1$ almost surely).
    Here, $\mbE[\cdot\,\,\vert\,\,\mcF^{[0, n]}]$ denotes the conditional expectation.
    Therefore, by dominated convergence, we have that 
    \begin{equation}
        \lim_{n\to\infty}
        \int_\Omega 
        d_\scrP\seq{
        V^{(n)}_\omega, 
        \operatorname{Proj}\seq{\StabMultDom{\Phi; \omega}}
        }
        \,\dee\mu(\omega)
        =
        0,
    \end{equation}
    where $V^{(n)} = \mbE\left[
           \operatorname{Proj}\seq{\StabMultDom{\Phi}}
           \,\,\vert\,\,
           \mcF^{[0, n]}
        \right]$.
    Because $T$ is invertible and $\mu$-preserving, we may rewrite this as 
    \begin{equation}\label{Eqn:Determinstic, eqn 1}
        \lim_{n\to\infty}
        \int_\Omega 
        d_\scrP\seq{
        V^{(n)}_{T^{-n-1}(\omega)}, 
        \operatorname{Proj}\seq{\mcM^\infty_{\Phi; T^{-n-1}(\omega)}}
        }
        \,\dee\mu(\omega)
        =
        0.
    \end{equation}
    On the other hand, by Corollary \ref{Cor:Unitary}, we know that
    \begin{equation}
        \Phi^{(n+1)}_{T^{-n-1}(\omega)}\seq{\StabMultDom{\Phi; T^{-n-1}(\omega)}} = \StabMultDom{\Phi; \omega}
    \end{equation}
    holds for $\mu$-almost every $\omega$ for all $n\in\mbN$.
    Therefore, because $\operatorname{Proj}$ is an isometry of metric spaces and because $\Phi^{(n+1)}_{T^{-n-1}(\omega)}$ acts isometrically on $\StabMultDom{\Phi; T^{-n-1}(\omega)}$ by Corollary \ref{Cor:Unitary} for $\mu$-almost every $\omega$, by Lemma \ref{Lem:App:Lipschitz for Grassmann}, we conclude that 
    \begin{align*}
        &\int_\Omega 
        d_{\mbG\seq{\matrices}}\seq{
        \Phi^{(n+1)}_{T^{-n-1}(\omega)}\seq{\operatorname{Proj}^{-1}\seq{V^{(n)}_{T^{-n-1}(\omega)}}}, 
        \mcM^\infty_{\Phi; \omega}
        }
        \,\dee\mu(\omega)\\
        &=\int_\Omega 
        d_{\mbG\seq{\matrices}}\seq{
        \Phi^{(n+1)}_{T^{-n-1}(\omega)}\seq{\operatorname{Proj}^{-1}\seq{V^{(n)}_{T^{-n-1}(\omega)}}}, 
        \Phi^{(n+1)}_{T^{-n-1}(\omega)}\seq{\mcM^\infty_{\Phi; T^{-n-1}(\omega)}}
        }
        \,\dee\mu(\omega)\\
        &\leq 
        \int_\Omega 
        d_{\mbG\seq{\matrices}}\seq{
        \operatorname{Proj}^{-1}\seq{V^{(n)}_{T^{-n-1}(\omega)}}, 
        \mcM^\infty_{T^{-n-1}(\omega)}
        }
        \,\dee\mu(\omega)\\
        &\leq 
        \int_\Omega 
        d_{\scrP}\seq{
        V^{(n)}_{T^{-n-1}(\omega)}, 
        \operatorname{Proj}\seq{\mcM^\infty_{T^{-n-1}(\omega)}}
        }
        \,\dee\mu(\omega).
    \end{align*}
    In particular, by (\ref{Eqn:Determinstic, eqn 1}), we conclude that 
    \begin{equation}\label{Eqn:Determinstic, eqn 2}
        \lim_{n\to\infty}
        \int_\Xi 
        d_{\mbG\seq{\matrices}}\seq{
        \Phi^{(n+1)}_{T^{-n-1}(\omega)}\seq{\operatorname{Proj}^{-1}\seq{V^{(n)}_{T^{-n-1}(\omega)}}}, 
        \mcM^\infty_{\omega}
        }
        \,\dee\mu(\omega)
        =
        0.
    \end{equation}
    Therefore, if we let $W^{(n)}:\Omega\to\scrP_\delta$ be the $\mbM_{d^2}$-valued random variable 
    \begin{equation}
        W^{(n)}_\omega
        =
        \operatorname{Proj}\seq{
         \Phi^{(n+1)}_{T^{-n-1}(\omega)}\seq{\operatorname{Proj}^{-1}\seq{V^{(n)}_{T^{-n-1}(\omega)}}}
        },
    \end{equation}
    then, because $V^{(n)}$ is $\mcF^{[0, n]}$-measurable and $\Phi^{(n+1)}\circ T^{-n-1}$ is $\mcF^{[-n-1, -1]}$-measurable for all $n$, we see that $W^{(n)}$ is $\mcF^{<0}$-measurable for all $n$, and, moreover, again by the isometry of Lemma \ref{Lem:App:Isometry_Grassmannian}, we have from (\ref{Eqn:Determinstic, eqn 2}) that 
    \begin{equation}
        \lim_{n\to\infty}
        \int_\Omega
        d_\scrP\seq{W^{(n)}_\omega, \StabMultDom{\Phi; \omega}}\,\dee\mu(\omega)
        =
        0.
    \end{equation}
    Therefore, there is a subsequence $\seq{n_k}_{k\in\mbN}$ such that 
    \begin{equation}
        \lim_{k\to\infty}
        d_\scrP\seq{W^{(n_k)}_\omega, \operatorname{Proj}\seq{\StabMultDom{\Phi; \omega}}}
        =
        0
    \end{equation}
    holds for $\mu$-almost every $\omega\in\Omega$. 
    In particular, we conclude that $\mcM' := \operatorname{Proj}^{-1}\seq{\lim_k W^{(n_k)}}\in\mcF^{<0}$ satisfies $\mcM'\in\mcF^{<0}$ and $\mcM' = \StabMultDom{\Phi}$ almost surely. 
    As discussed above, this was all we needed to show to prove the result, so the proof is concluded. 
\end{proof}
\begin{remark}
    The assumption \Indep\, will hold in the following general setup. 
    Assume that $\Omega = \bcp^\mbZ$ and that $\mcF$ is the $\sigma$-algebra generated by cylinder sets, where $\bcp$ is given the Borel $\sigma$-algebra per usual. 
    We define $T:\Omega\to\Omega$ to be the shift map, $T(\seq{\psi_n}_{n\in\mbZ}) = \seq{\psi_{n+1}}_{n\in\mbZ}$, and we let $\mu$ be any measure on $\Omega$ that is invariant and ergodic under $T$. 
    If we define $\phi:\Omega\to\bcp$ by $\phi_{\seq{\psi_{n}}_{n\in\mbZ}} = \psi_0$, we see that \Indep\, is satisfied,
    even though it need not be (at least \textit{a priori}) that the random variables $\seq{\phi_n}_{n\in\mbZ}$ are independent.
\end{remark}
\begin{example}[Nonconstant \texorpdfstring{$\StabMultDom{\Phi}$}{l}]\label{Ex:Nondeterministic}
    The following example shows that $\StabMultDom{\Phi}$ need not be deterministic in general. 
    Let $\Omega = [0, 1)$, let $\mu$ be the Lebesgue measure, and let $\mcF$ be the Borel $\sigma$-algebra.
    Let $\alpha\in[0, 1)$ be an irrational number, and let $T:\Omega\to\Omega$ be the quasiperiodic rotation map $T(\omega) = \omega + \alpha - \lfloor\omega + \alpha \rfloor$, where $\lfloor\cdot\rfloor$ denotes the floor function. 
    Then it is a standard fact that $T$ is an invertible and ergodic map for $\mu$. 
    Define vectors in $\mbC^2$ by 
    \begin{equation}
        |v_\omega\rangle
       =
        \seq{\begin{array}{cc}
             \sin(2\pi \omega)  \\
             \cos(2\pi \omega)
        \end{array}}
        \qquad\text{and}\qquad 
        |w_\omega\rangle
       =
        \seq{\begin{array}{cc}
             \cos(2\pi \omega)  \\
             -\sin(2\pi \omega)
        \end{array}}
    \end{equation}
    for all $\omega\in\Omega$, where we have used the standard quantum mechanics notation to denote vectors.
    Note that $\langle w_\omega|v_\omega\rangle = 0$. 
    We define $\phi:\Omega\to\bcp$ by 
    \begin{equation}
        \phi_\omega(a)
        =
        \ketbra{v_{T(\omega)}}{v_\omega}a\ketbra{v_\omega}{v_{T(\omega)}}
        +
        \ketbra{w_{T(\omega)}}{w_\omega}a\ketbra{w_\omega}{w_{T(\omega)}}\quad a\in\matrices.
    \end{equation}
    Then notice that $\Phin_\omega(a) =  \ketbra{v_{T^{n-1}(\omega)}}{v_\omega}a\ketbra{v_\omega}{v_{T^{n-1}(\omega)}}
        +
        \ketbra{w_{T^{n-1}(\omega)}}{w_\omega}a\ketbra{w_\omega}{w_{T^{n-1}(\omega)}}$
    for all $n\in\mbN$. 
    Therefore, 
    \begin{equation}
        \Phi^{(n)*}_\omega(a) =  \ketbra{v_\omega}{v_{T^{n-1}(\omega)}}a\ketbra{v_{T^{n-1}(\omega)}}{v_\omega}
        +
        \ketbra{w_\omega}{w_{T^{n-1}(\omega)}}a\ketbra{w_{T^{n-1}(\omega)}}{w_\omega},
    \end{equation}
    so, because by \cite[Lemma 2.1]{Rahaman2017MultiplicativeChannels} $\mcM_{\Phin_\omega}$ is equal to the set of fixed points of $\Phin_\omega\circ\Phi^{(n)*}_\omega$, we see that $\mcM_{\Phin_\omega}$ is the $C^*$-algebra $\mcA_\omega$ generated by $\set{\ketbra{v_\omega}{v_\omega}, \ketbra{w_\omega}{w_\omega}}$ for all $n$. 
    Hence, $\StabMultDom{\Phi; \omega} = \mcA_\omega$. 
    If, however, $\mcA_\omega = \mcA_{T(\omega)}$ almost surely, then we would have that $\ketbra{v_{T(\omega)}}{v_{T(\omega)}}\in\mcA_\omega$ almost surely. 
    However, because $\ketbra{v_\omega}{v_\omega}$ and $\ketbra{w_\omega}{w_\omega}$ are orthogonal, the only rank one projections in $\mcA_\omega$ are exactly $\ketbra{v_\omega}{v_\omega}$ and $\ketbra{w_\omega}{w_\omega}$. 
    So, $\ketbra{v_{T(\omega)}}{v_{T(\omega)}}\in\mcA_\omega$ is almost surely false, because $\ket{v_\omega}\neq \ket{v_{T(\omega)}}$ and $\ket{w_\omega}\neq \ket{v_{T(\omega)}}$ almost surely. 
    Therefore, $\mcA_\omega = \StabMultDom{\Phi; \omega}$ is nondeterministic. 
\end{example}
We now turn our attention towards proving Theorem \ref{Thm:Main theorem}. 
We begin by proving part (c) of this theorem, which by Theorem \ref{Thm:Deterministic}, follows in a straightforward way. 
\begin{prop}
    Assume \BCPa\, and \Indep.
    Then Theorem \ref{Thm:Main theorem} (c) holds. 
\end{prop}
\begin{proof}
    Note that by Theorem \ref{Thm:Deterministic} and Corollary \ref{Cor:Unitary}, the map $\Phi^{(\tau)}_\omega\vert_{\mcA_\Phi}:\mcA_\Phi\to\mcA_\Phi$ is unitary with respect to the Hilbert-Schmidt inner product on the $\mcA_\Phi$. 
    Therefore, by the spectral theorem applied to $\Phi^{(\tau)}_\omega\vert_{\mcA_\Phi}$, there is an orthonormal basis $\set{u_\lambda}_{\lambda\in\Sigma_\omega}$ of $\mcA_\Phi$, where $\Sigma_\omega$ is the spectrum of $\Phi^{(\tau)}_\omega\vert_{\mcA_\Phi}$ and each $u_\lambda\in\matrices$ satisfies $\Phi^{(\tau)}(u_\lambda) = \lambda u_\lambda$. 
    Because $\Phi^{(\tau)}_\omega\vert_{\mcA_\Phi}$ is unitary, we know that $\Sigma_\omega\subset\mbT$ as well. 
    In particular, we see that $\mcA_\Phi$ is contained in the $C^*$-algebra generated by the set $\set{a\in\matrices\,\,:\,\, \Phi^{(\tau)}_\omega(a) = \lambda a\text{ for some $\lambda\in\mbT$}}$. 
    On the other hand, if $a$ satisfies $\Phi^{(\tau)}_\omega(a) = \lambda a$ for some $\lambda\in\mbT$, then it is clear that $\hsnorm{a} = \hsnorm{\Phi^{(\tau)}_\omega(a)}$, and therefore by Lemma \ref{Lem:Mult dom of bcp map} any such $a$ is contained in $\mcA_\Phi = \mcM_{\Phitau{\omega}}$. 
    Therefore, because $A_\Phi$ is a $C^*$-algebra, the $C^*$-algebra generated by all such $a$ is necessarily contained in $A_\Phi$, so, by the above, we conclude (c), and end the proof. 
\end{proof}
\begin{remark}
    Notice that if $a\in\matrices$ is such that $\Phi^{(\tau)}_\omega(a) = \lambda a$, then by Lemma \ref{Lem:Mult dom of bcp map} we have that $a\in\mcM_{\Phi^{(\tau)}_\omega}$. 
    Therefore, the $C^*$-algebra $\mcA_{\omega}$ generated by all such $a$ is always contained in $\mcM_{\Phi^{(\tau)}_\omega} = \StabMultDom{\Phi; \omega}$. 
    However, as Example \ref{Ex:Nondeterministic} shows, the equality $\mcA_\omega = \StabMultDom{\Phi; \omega}$ need not hold in general. 
\end{remark}
Next, to prove part (a) of Theorem \ref{Thm:Main theorem}, we need the following two technical lemmas. 
\begin{lemma}\label{Cor:Linear algebra corollary}
For almost every $\omega\in\Omega$, $\phi_\omega\seq{\StabMultDomPerp{\Phi; \omega}}\subseteq\StabMultDomPerp{\Phi; T(\omega)}$ 
    and $\phi_\omega^*\seq{\StabMultDomPerp{\Phi; T(\omega)}}\subseteq\StabMultDomPerp{\Phi; \omega}$.
\end{lemma}
\begin{proof}
    From Corollary \ref{Cor:Unitary}, we know that $\phi_\omega\seq{\LocalMultCore{\omega}}=\LocalMultCore{T(\omega)}$ 
    and 
    $\phi_\omega^*\seq{\LocalMultCore{T(\omega)}}=\LocalMultCore{\omega}$.
    So, let $a\in \StabMultDomPerp{\Phi; \omega}$, and let $b\in\LocalMultCore{T(\omega)}$. 
    Then 
    \begin{equation}
        \innerhs{\phi_\omega(a)}{b}
        =
        \innerhs{a}{\phi_\omega^*(b)}
        =
        0, 
    \end{equation}
    since $\phi_\omega^*(b)\in\LocalMultCore{\omega}$. 
    So, $\phi_\omega\seq{\StabMultDomPerp{\Phi; \omega}}\subseteq\StabMultDomPerp{\Phi; T(\omega)}$. 
    Arguing the same way yields the containment  $\phi_\omega^*\seq{\StabMultDomPerp{\Phi; T(\omega)}}\subseteq\StabMultDomPerp{\Phi; \omega}$, and the proof is done.
\end{proof}
\begin{lemma}\label{Lem:Kingman technical}
    The inequality
    \begin{equation}
        \inf_n \frac{1}{n}\int_\Omega \log \hsnorm{\Phi^{(n)}_\omega\vert_{\StabMultDomPerp{\Phi; \omega}}}\,
        \dee\mu(\omega)
        < 
        0
    \end{equation}  
    holds. 
\end{lemma}
\begin{proof}   
    By applying Lemma \ref{Lem:Mult dom of bcp map} to the bcp map $\Phi^{(\tau)}_\omega$, we have that 
    \begin{equation}\label{Eqn:Perp space as shrinking}
         \mcM_{\Phi^{(\tau)}_\omega}^\perp
         \setminus\{0\}
            \subset
        \set{a\in\matrices\,\,:\,\,
            \|a\|_{\operatorname{HS}}
            >
            \|\Phi^{(\tau)}_\omega(a)\|_{\operatorname{HS}}
        }
    \end{equation}
    for almost every $\omega\in\Omega$. 
    Since $\mcM_{\Phi^{(\tau)}_\omega}^\perp\cap \set{a\in\matrices\,\,:\,\,\|a\|_{\operatorname{HS}} = 1}$ is compact, there is $a_{\operatorname{max}}\in \mcM_{\Phi^{(\tau)}_\omega}^\perp$ such that $\hsnorm{a_{\operatorname{max}}} = 1$ and 
    $\hsnorm{\Phitau{\omega}(a_{\operatorname{max}})}
        =
        \hsnorm{\Phitau{\omega}\vert_{\mcM^\perp_{\Phitau{\omega}}}
        }$,
    so (\ref{Eqn:Perp space as shrinking}) gives
    \begin{equation}\label{Eqn:Strict inequality for phi tau}
        \hsnorm{\Phitau{\omega}\vert_{\mcM^\perp_{\Phitau{\omega}}}}
        <
        1
    \end{equation}  
    for almost every $\omega\in\Omega$. 
    Now, because $\tau(\omega) < \infty$ for almost every $\omega\in\Omega$, there is $N\in\mbN$ such that $\mu[\tau = N] > 0$. 
    Fix such an $N$. 
    Because $\hsnorm{\Phi^{(N)}}\leq 1$ almost everywhere (from Lemma \ref{Lem:Mult dom of bcp map}), noting that $\mcM_{\Phitau{\omega}} = \StabMultDom{\Phi; \omega}$, we may conclude from (\ref{Eqn:Strict inequality for phi tau}) that
    \begin{align*}
        \int_\Omega 
         \log \left\|
            \Phi^{(N)}_\omega\vert_{\StabMultDomPerp{\Phi; \omega}}
        \right\|_{\operatorname{HS}}\,
        \dee\mu(\omega)
        &\leq 
       \int_{\tau = N} 
       \log \left\|
            \Phi^{(N)}_\omega\vert_{\StabMultDomPerp{\Phi; \omega}}
        \right\|_{\operatorname{HS}}\,
        \dee\mu(\omega)\\
        &= 
        \int_{\tau = N} 
       \log \left\|
            \Phi^{(\tau)}_\omega\vert_{\mcM_{\Phi^{(\tau)}_\omega}^\perp}
        \right\|_{\operatorname{HS}}\,
        \dee\mu(\omega)
        < 
        0,
    \end{align*}
    which concludes the proof. 
\end{proof}
We now have the tools we need to prove the more general version of Theorem \ref{Thm:Main theorem} that doesn't require \Indep, from which it will be clear that Theorem \ref{Thm:Main theorem} (a) and (b) follow. 
Since we have already proved Theorem \ref{Thm:Main theorem} (c) above, after we prove the following theorem, our treatment of Theorems \ref{Thm:Deterministic} and \ref{Thm:Main theorem} will be concluded. 
\begin{thm}\label{Thm:More general main theorem}
    Assume \BCPa. 
    \begin{enumerate}[label = (\alph*)]
        \item For almost every $\omega\in\Omega$, the equality 
        \begin{equation}
            \StabMultDomPerp{\Phi; \omega} = V^{\leq\lambda_2}_\omega
        \end{equation}
        holds, where $\lambda_2<0$ is the second Lyapunov exponent associated to the linear cocycle $(T, \phi)$ and $V^{\leq\lambda_2}$ is the corresponding Lyapunov subspace. 

        \item For almost every $\omega\in\Omega$, $\phi_\omega\seq{\StabMultDom{\Phi; \omega}} = \StabMultDom{\Phi; T(\omega)}$, and $\phi_\omega\vert_{\StabMultDom{\Phi; \omega}}:\StabMultDom{\Phi; \omega} \to \StabMultDom{\Phi; T(\omega)}$ defines a $*$-isomorphism of $C^*$-algebras with inverse $\phi^*\vert_{\StabMultDom{\Phi; \omega}}$.
    \end{enumerate}
\end{thm}
\begin{proof}
    Part (b) is proved in light of Corollary \ref{Cor:Unitary}, so all that needs to be done is to prove (a). 
    To do this, note that by Theorem \ref{Thm:Deterministic}, it suffices to show that both inclusions $\StabMultDomPerp{\Phi; \omega}\supseteq V^{\leq\lambda_2}_\omega$ and $\StabMultDomPerp{\Phi; \omega}\subseteq V^{\leq\lambda_2}_\omega$ hold for $\mu$-almost every $\omega\in\Omega$. 
    We begin by showing the first inclusion. 
    First, notice that almost every $\omega\in\Omega$ and for any $b\in V^{\leq\lambda_2}_\omega$, we have that $\lim_{n}n^{-1}\log\hsnorm{\Phin_\omega(b)} < 0$, so in particular
    \begin{equation}
        \lim_{n\to\infty}\hsnorm{\Phin_\omega(b)} = 0.
    \end{equation}
    Therefore, if we let $a\in\StabMultDom{\Phi;\omega}$, we have that 
    \begin{align}
        |\innerhs{a}{b}|
            =
        |\tr{a^*b}|
            &=
        \lim_{n\to\infty}
        \left|\tr{\Phin_\omega(a^*b)}\right|
        &&\text{$\Phin_\omega$ is trace preserving}\notag\\
            &=
        \lim_{n\to\infty}
        \left|\tr{\Phin_\omega(a)^*\Phin_\omega(b)}\right|
            &&\text{since $a^*\in\mcM_{\Phin_\omega}$ for all $n$}\notag\\
            &\leq 
        \lim_{n\to\infty}
            \|a\|_{\operatorname{HS}}^2
            \|\Phin_\omega(b)\|_{\operatorname{HS}}^2
            &&\text{by Cauchy-Schwarz and $\|\Phin_\omega\|_{\operatorname{HS}}\leq 1$}\notag\\
            &= 0,\label{Eq:Calculation_pf_main_thm} 
    \end{align}
    which establishes that $\StabMultDomPerp{\Phi; \omega}\supseteq V^{\leq\lambda_2}_\omega$.
    Next, we show that $\StabMultDomPerp{\Phi; \omega}\subseteq V^{\leq\lambda_2}_\omega$. 
    Let $F_n:\Omega\to[-\infty, 0]$ be the function defined by 
    \begin{equation}
        F_n(\omega)
        =
        \log\left\|\Phin_\omega\vert_{\StabMultDomPerp{\Phi; \omega}}\right\|_{\operatorname{HS}},
    \end{equation}
    where we take $\log(0):= - \infty$.
    Now, for $n, m\in\mbN$, note that $\Phi^{(n+m)}_\omega = \Phi^{(m)}_{T^n(\omega)}\circ\Phin_\omega$, so from Corollary \ref{Cor:Linear algebra corollary} we have  
    \begin{align*}
    \left\|\Phi^{(n +m)}_\omega\vert_{\StabMultDomPerp{\Phi; \omega}}\right\|_{\operatorname{HS}}
        &= 
    \sup 
        _{a\in \StabMultDomPerp{\Phi; \omega}: \|a\|_{\operatorname{HS}} = 1}
        \left\|
            \Phi^{(n+m)}_\omega(a)
        \right\|_{\operatorname{HS}}\\
        &\leq
        \left\|
            \Phi^{(m)}_{T^n(\omega)}
            \vert_{\StabMultDomPerp{\Phi; T^n(\omega)}}
        \right\|_{\operatorname{HS}}
        \cdot 
    \sup 
        _{a\in \StabMultDomPerp{\Phi; \omega}: \|a\|_{\operatorname{HS}} = 1}
        \left\|
            \Phi^{(n)}_\omega(a)
        \right\|_{\operatorname{HS}}\\
        &= 
    \left\|
            \Phi^{(m)}_{T^n(\omega)}
            \vert_{\StabMultDomPerp{\Phi; T^n(\omega)}}
        \right\|_{\operatorname{HS}}
    \left\|
            \Phi^{(n)}_{\omega}
            \vert_{\StabMultDomPerp{\Phi; \omega}}
        \right\|_{\operatorname{HS}}
    \end{align*}
    hence 
    \begin{equation}
        F_{n + m}(\omega)
        \leq 
        F_n(\omega)
        +
        F_m(T^n(\omega))
    \end{equation}
    for almost every $\omega\in\Omega$.  
    Clearly, $\max(F_1, 0) = 0 \in L^1(\Omega, \mu)$, so we may apply the Subadditive Ergodic Theorem \cite{Kingman1968TheProcesses, Kingman1973SubadditiveTheory} to $\seq{F_n}_{n\in\mbN}$ to conclude that there is $\kappa'\in [-\infty, 0]$ such that 
    \begin{equation}\label{Eqn:Kingman}
        \kappa'
            =
        \lim_{n\to\infty}\frac{1}{n}
        F_n(\omega)
            =
        \inf_n\frac{1}{n}\int_\Omega F_n(\omega')\dee\mu(\omega')
    \end{equation}
    for almost every $\omega\in\Omega$. 
    By Lemma \ref{Lem:Kingman technical}, we know that $\kappa' < 0$.
    If $\kappa' = -\infty$, then let $\kappa  = \lambda_2$; otherwise, let $\kappa = \kappa'$.
    In either case, because $\kappa'<0$, we know from the MET that $\kappa\leq\lambda_2$. 
    Thus, from (\ref{Eqn:Kingman}) we see that
    \begin{equation}
        \StabMultDomPerp{\Phi; \omega}
            \subseteq
        \set{
        a\in\matrices\,\,:\,\,
        \lim_{n\to\infty}
        \frac{1}{n}
        \log\|\Phi^{(n)}_\omega(a)\|_{\operatorname{HS}} < \kappa
        }
        \subseteq V^{\leq\lambda_2}_\omega
    \end{equation}
    holds for $\mu$-almost every $\omega\in\Omega$. 
    So, the proof of (a) is concluded, which ends the proof. 
\end{proof}
\subsection{Entanglement breaking}\label{Subsec:EB}
We now apply the theorems proved above to address \ENT. 
Recall the definition of entanglement breaking from the introduction.
\begin{definition}[Entanglement breaking]
    A linear map $\psi:\matrices\to\matrices$ is called entanglement breaking if, for any $k\in\mbN$ and any positive semidefinite matrix $\rho\in\matrices\otimes \mbM_k$, the positive semidefinite matrix $\psi\otimes\operatorname{Id}_{\mbM_k}(\rho)$ is separable, meaning that $\psi\otimes\operatorname{Id}_{\mbM_k}(\rho)$ belongs to the convex cone of $\matrices\otimes \mbM_k$ generated by the elements $P\otimes Q$ where $P\in\matrices$ and $Q\in\mbM_k$ are positive semidefinite.
    We let $\eb$ denote the set of all entanglement breaking maps $\matrices\to\matrices$.
\end{definition}
In \cite{Horodecki2003EntanglementChannels}, it was shown that entanglement breaking has the following equivalent formulations. 
\begin{thmx}[\texorpdfstring{\cite{Horodecki2003EntanglementChannels}}{l}]\label{Thm:EB}
    For a linear map $\psi:\matrices\to\matrices$, the following are equivalent. 
    \begin{enumerate}[label = (\alph*)]
        \item $\psi$ is entanglement breaking.

        \item For any positive map $\varphi:\matrices\to\matrices$, the maps $\varphi\circ\psi$ and $\psi\circ\varphi$ are completely positive. 

        \item For any completely positive map $\varphi:\matrices\to\matrices$, the maps $\varphi\circ\psi$ and $\psi\circ\varphi$ are entanglement breaking. 
        
        \item There exist states $\rho_j$ and positive semidefinite matrices $Q_j$ for $j = 1, \dots, k$ such that for any $a\in\matrices$, we have that 
        \begin{equation}
        \psi(a)
            =
        \sum_{j=1}^k
        \innerhs{\rho_j}{a}Q_j.
        \end{equation}
    \end{enumerate}
\end{thmx}
Note that $\eb$ is a closed subset of $\scrL$, which follows immediately from part (b) in the theorem above, as the set of completely positive maps is closed. 
Thus, the function 
\begin{equation}
    \begin{split}
        d_{\operatorname{HS}}(\cdot, \eb):\bcp &\to [0, \infty)\\
        d_{\operatorname{HS}}(\psi, \eb)
            &= 
        \inf_{\varphi\in\eb}\hsnorm{\psi - \varphi}
    \end{split}
\end{equation}
is continuous. 
Our theorems concern the concept of \textit{asymptotic} entanglement breaking, whose definition we recall. 
\begin{definition}[Asymptotic entanglement breaking]
    Let $\Psi = \seq{\psi_n}_{n\in\mbZ}$ be a bi-infinite sequence of linear maps $\psi_n:\matrices\to\matrices$.
    If 
    \begin{equation}
        \lim_{n\to\infty}d_{\operatorname{HS}}\seq{\Psi^{(n)}, \eb}
        =
        0,
    \end{equation}
    where $\Psi^{(n)} = \psi_{n-1}\circ\cdots\circ\psi_0$,
    we say that $\Psi$ is asymptotically entanglement breaking (a.e.b.). 
    We call $\Phi$ a.e.b. at $\omega\in\Omega$ if $\Phi_\omega$ is a.e.b., and we say that $\Phi$ is a.e.b. if $\Phi_\omega$ is a.e.b. for almost every $\omega\in\Omega$. 
\end{definition}
To prove our theorems regarding asymptotic entanglement breaking, it is appropriate to recall some relevant results. 
The following is established in \cite[Corollary 3]{Strmer2008SeparableMaps} and \cite[Theorem 10]{Strmer2008SeparableMaps}, noting that $\psi\vert_{\mcM_\psi}$ is injective whenever $\psi$ is a bcp map. 
\begin{thmx}[\texorpdfstring{\cite{Strmer2008SeparableMaps}}{l}]
\label{Thm:Stormer}
    Let $\psi:\matrices\to\matrices$ be a bcp map. 
    If $\psi$ is entanglement breaking, then $\mcM_\psi$ is an abelian $C^*$-algebra. 
    If, on the other hand, $\psi(\matrices)$ is contained in an abelian $C^*$-algebra, then $\psi$ is entanglement breaking. 
\end{thmx}
Towards proving Theorem \ref{Thm:Abelian}, we start with the following lemma. 
We call $\psi\in\bcp$ a cluster point of $\Phi_\omega$ if $\psi$ is a cluster point of the forward-time sequence $\seq{\phi_{n; \omega}}_{n\geq 0}$.
\begin{lemma}\label{Lem:Cluster points}
    For almost every $\omega\in\Omega$, there are cluster points of $\Phi_\omega$, which are necessarily in $\bcp$. 
    Moreover, for any cluster point $\psi\in\bcp$ of $\Phi_\omega$, we have that $\mcM_\psi = \StabMultDom{\Phi; \omega}$. 
\end{lemma}
\begin{proof}
    The first statement follows because $\bcp$ is a compact subset of $\scrL$ and $\Phi_\omega\subset\bcp$. 
    To see the next part, let $\psi\in\bcp$ be a cluster point of $\Phi_\omega$, and let $\Phi^{(n_k)}_\omega\to\psi$. 
    First, we show that $\StabMultDom{\Phi; \omega}\subseteq \mcM_\psi$. 
    Let $a\in\StabMultDom{\Phi; \omega}$, so $a\in\mcM_{\Phi^{(n_k)}_\omega}$ for all $k$.
    Then by Lemma \ref{Lem:Mult dom of bcp map}, we have that 
    \begin{equation}
        \hsnorm{\psi(a)}
        =
        \lim_{k\to\infty}
        \hsnorm{\Phi^{(n_k)}_\omega(a)}
        =
        \lim_{k\to\infty}\hsnorm{a}
        =
        \hsnorm{a}.
    \end{equation}
    So, again by Lemma \ref{Lem:Mult dom of bcp map} we conclude that $a\in\mcM_\psi$. 
    Thus, $\StabMultDom{\Phi; \omega}\subseteq\mcM_\psi$. 
    On the other hand, suppose that $a\in\mcM_\psi\cap \StabMultDomPerp{\Phi; \omega}$.
    Then by Theorem \ref{Thm:More general main theorem}, we know that there is $\kappa<0$ such that 
    \begin{equation}
        \lim_{k\to\infty}
        \frac{1}{k}
        \log\hsnorm{\Phi^{(n_k)}_\omega(a)}
        =
        \kappa.
    \end{equation}
    In particular, $\lim_k\hsnorm{\Phi^{(n_k)}_\omega(a)} = 0$. 
    Thus, by Lemma \ref{Lem:Mult dom of bcp map}, we conclude that 
    \begin{equation}
    \hsnorm{a}
        =
    \hsnorm{\psi(a)}
        =
    \lim_{k\to\infty}
    \hsnorm{\Phi^{(n_k)}_\omega(a)} 
        = 
    0,
    \end{equation}
    so $a = 0$. 
    Thus, $\mcM_\psi\cap \StabMultDomPerp{\Phi; \omega} = \set{0}$, so since $\StabMultDom{\Phi; \omega}\subseteq\mcM_\psi$, we conclude that $\mcM_\psi = \StabMultDom{\Phi; \omega}$, which ends the proof. 
\end{proof}
Next, we have another simple lemma. 
\begin{lemma}\label{Lem:psi(matrices) = psi(mcMpsi)}
    For almost every $\omega\in\Omega$, any cluster point $\psi$ of $\Phi_\omega$ satisfies $\psi(\matrices) = \psi(\mcM_\psi)$. 
\end{lemma}
\begin{proof}
    Let $\psi$ be a cluster point of $\Phi_\omega$, and let $\Phi^{(n_k)}_\omega\to\psi$. 
    By the previous lemma, we know that $\mcM_\psi = \StabMultDom{\Phi; \omega}$. 
    Thus, by the decomposition $\matrices = \StabMultDom{\Phi; \omega}\oplus\StabMultDomPerp{\Phi; \omega}$, any $a\in\matrices$ may be written as $a = b + c$ with $b\in\mcM_\psi$ and $c\in\StabMultDomPerp{\Phi; \omega}$.
    From Theorem \ref{Thm:More general main theorem}, we know that 
    \begin{equation}
        \hsnorm{\psi(c)}
        =
        \lim_{k\to\infty}
        \hsnorm{\Phi^{(n_k)}_\omega(c)}
        =
        0.
    \end{equation}
    Thus, we have that $\psi(a) = \psi(b)$. 
    Because $a\in\matrices$ was arbitrary, we see that $\psi\seq{\matrices}=\psi(\mcM_\psi)$, which concludes the proof. 
\end{proof}
This allows us to conclude the main technical result required for proving (a generalization of) Theorem \ref{Thm:Abelian}.
\begin{lemma}\label{Lem:Asym eb iff cluster point}
    For almost every $\omega\in\Omega$, the following are equivalent. 
    \begin{enumerate}[label = (\alph*)]
        \item $\Phi$ is asymptotically entanglement breaking at $\omega$. 
        
        \item There exists a cluster point of $\Phi_\omega$ in $\eb$. 

        \item Every cluster point of $\Phi_\omega$ is in $\eb$. 
    \end{enumerate}
\end{lemma}
\begin{proof}
    Assume (a) and let $\psi$ be a cluster point of $\Phi_\omega$. 
    Then there is a subsequence $\seq{\Phi^{(n_k)}_\omega}_{k\in\mbN}$ such that $\Phi^{(n_k)}_\omega\to\psi$. 
    On the other hand, we have that $d_{\operatorname{HS}}\seq{\Phi^{(n_k)}_\omega, \eb}\to 0.$
    Thus, because $d_{\operatorname{HS}}(\cdot, \eb)$ is continuous and $\eb$ is closed, it must be that $d_{\operatorname{HS}}(\psi, \eb) = 0$, hence $\psi\in\eb$. 
    Conversely, to see that (c) implies (a), it suffices to show that for any subsequence $\seq{n_k}_{k\in\mbN}\subset\mbN$, there exists a further subsequence $\seq{n_{k_j}}_{j\in\mbN}\subset\seq{n_k}_{k\in\mbN}$ such that 
    \begin{equation}
        \lim_{j\to\infty}
        d_{\operatorname{HS}}\seq{\Phi^{\seq{n_{k_j}}}_\omega, \eb}
        =
        0.
    \end{equation}
    To see this, note that for any subsequence $\seq{n_k}_{k\in\mbN}$, the set $\seq{\Phi^{(n_k)}_\omega}_{k\in\mbN}$ is a subset of the compact set $\bcp$. 
    Thus, there exists a convergent subsequence $\seq{\Phi^{\seq{n_{k_j}}}_\omega}_{j\in\mbN}$, converging to some $\psi\in\eb$ by assumption (c). 
    Because $\psi\in\eb$, we have in particular that 
    \begin{equation}
        \lim_{j\to\infty}
        d_{\operatorname{HS}}\seq{\Phi^{\seq{n_{k_j}}}_\omega, \eb}
        =
        0,
    \end{equation}
    which by the above discussion allows us to conclude that (c) is true. 
    Of course, (c) immediately implies (b), so it remains to show that (b) implies (c). 
    So, assume (b). 
    By Lemma \ref{Lem:Cluster points} and Theorem \ref{Thm:Stormer}, $\Phi_\omega$ having an entanglement breaking cluster point implies that $\StabMultDom{\Phi; \omega}$ is abelian. 
    So, if we let $\psi\in\bcp$ be an arbitrary cluster point of $\Phi_\omega$, then another application of Lemma \ref{Lem:Cluster points} shows that $\mcM_\psi$ is abelian. 
    By Lemma \ref{Lem:psi(matrices) = psi(mcMpsi)}, we know that $\psi(\matrices) = \psi(\mcM_\psi)$.
    Now, because $\psi\vert_{\mcM_\psi}$ is a $*$-homomorphism, and because $\mcM_\psi$ is a $C^*$-algebra, we know that $\psi\seq{\mcM_\psi}$ is a $C^*$-algebra \cite[Theorem 3.1.6]{Murphy2007C-AlgebrasTheory}.
    So, because $\mcM_\psi$ is abelian, we know that $\psi\seq{\mcM_\psi}$ is also abelian. 
    Therefore, $\psi\seq{\matrices} = \psi\seq{\mcM_\psi}$ is an abelian $C^*$-algebra, so by Theorem \ref{Thm:Stormer}, we conclude that $\psi\in\eb$. 
    Since $\psi$ was an arbitrary cluster point of $\Phi_\omega$, the proof is concluded. 
\end{proof}
We now have enough to prove the appropriate generalization of Theorem \ref{Thm:Abelian} without the assumption \Indep. 
Recall that $X_{\operatorname{aeb}} = \set{\omega\in\Omega\,\,:\,\, \Phi_\omega \text{ is a.e.b. at }\omega}$.
\begin{thm}\label{Thm:General Abelian}
    Assume \BCPa. 
    Then $\mu[X_{\operatorname{aeb}}]\in\set{0, 1}$. 
    Moreover, $\mu[X_{\operatorname{aeb}}] = 1$ if and only if $\StabMultDom{\Phi}$ is abelian with positive probability. 
\end{thm}
\begin{proof}
    We begin by showing that $\mu[X_{\operatorname{aeb}}]\in\set{0, 1}$. 
    %
    % To do this, it suffices to show that $\mu[X_{\operatorname{aeb}}]>0$ implies $\mu[X_{\operatorname{aeb}}]=1$. 
    % %
    % So, assume $\mu[X_{\operatorname{aeb}}]>0$.
    % %
    By Lemma \ref{Lem:Asym eb iff cluster point}, we know that
    \begin{equation}
        \mu[X_{\operatorname{aeb}}]
        =
        \mu[
        \omega\in\Omega\,\,:\,\, \text{there is a cluster point of $\Phi_\omega$ in $\eb$}
        ].
    \end{equation}
    If, however, $\omega\in\Omega$ is such that $\Phi_\omega$ has a cluster point $\psi\in\eb$, then by Theorem \ref{Thm:EB}, we see that $\psi' = \psi\circ\phi_{T^{-1}(\omega)}\in\eb$ and  $\psi'$ is a cluster point of $\Phi_{T^{-1}(\omega)}$. 
    Therefore, by Proposition \ref{Prop:Basics_of_ergodicity}, we conclude that 
    \begin{equation}
         \mu[
        \omega\in\Omega\,\,:\,\, \text{there is a cluster point of $\Phi_\omega$ in $\eb$}
        ]
        \in
        \set{0,1}, 
    \end{equation}
    which is what we wanted. 
    Next, we show that $\mu[X_{\operatorname{aeb}}] = 1$ is equivalent to $\StabMultDom{\Phi}$ being abelian with positive probability. 
    First, assume $\mu[X_{\operatorname{aeb}}] = 1$. 
    Then by Lemma \ref{Lem:Cluster points} we see that $\StabMultDom{\Phi; \omega} = \mcM_\psi$ almost surely, where $\psi$ is any cluster point of $\Phi_\omega$. 
    But by Theorem \ref{Thm:Stormer} and Lemma \ref{Lem:Asym eb iff cluster point}, we have that $\psi\in\eb$ hence $\mcM_\psi$ is abelian. 
    Therefore, on the event $X_{\operatorname{aeb}}$, we have that $\StabMultDom{\Phi}$ is abelian. 
    Because $\mu[X_{\operatorname{aeb}}] = 1$, this occurs with positive probability. 
    Conversely, if $\StabMultDom{\Phi}$ is abelian with positive probability, then again by Theorem \ref{Thm:Stormer}, Lemma \ref{Lem:Cluster points}, and Lemma \ref{Lem:Asym eb iff cluster point}, we have that $\mu[X_{\operatorname{aeb}}]>0$. 
    But we have already seen that $\mu[X_{\operatorname{aeb}}]\in\set{0, 1}$, hence $\mu[X_{\operatorname{aeb}}] = 1$, which concludes the proof. 
\end{proof}
We may now conclude the corollary of this result advertised in the introduction above, which we recall now. 
\PPT*
\begin{proof}
    It was shown in \cite{Rahaman2018EventuallyMaps} that any PPT bcp map $\psi:\matrices\to\matrices$ satisfies that $\mcM_\psi$ is abelian. 
    Thus, because $\phi_\omega$ is PPT with positive probability, the almost sure inclusion
    \begin{equation}
        \mcM_{\phi_\omega}\supseteq \StabMultDom{\Phi; \omega}
    \end{equation}
    implies that $\StabMultDom{\Phi; \omega}$ is abelian almost surely.
    Thus, by Theorem \ref{Thm:Abelian}, we may conclude that $\Phi$ is asymptotically entanglement breaking, as desired. 
\end{proof}
The following example shows that the above theorem is not implied by the $\mathrm{PPT}^2$ conjecture. 
Better yet, we show that there are ergodic quantum processes $\Phi$ such that $\mu[\phi\text{ is PPT}]>0$ while 
\begin{equation}
    \mu[
    \text{only one of }\phi_1, \phi_2, \dots, \phi_n\text{ is PPT}
    ]
    =
    1
\end{equation}
for arbitrary $n\in\mbN$.
\begin{example}[Composition structure defined by graphs]\label{Ex:Markovian}
    Let $\Gamma = \seq{V_\Gamma, E_\Gamma}$ be a strongly connected directed finite graph, where $V_\Gamma$ denotes the finite vertex set and $E_\Gamma$ denotes the finite edge set. 
    For an edge $e$, denote it by $e = \seq{v, w}$ where $e$ is the edge going from $v$ to $w$. 
    Let $\vec{p} = \seq{p_{e}}_{e\in E_{\Gamma}}$ be a set of numbers in $(0, 1]$ such that 
    \begin{equation}
        \sum_{w\in V_\Gamma\,\,:\,\, \seq{v, w}\in E_\Gamma}
            p_{\seq{v, w}}
                =
            1.
    \end{equation}
    Because $\Gamma$ is strongly connected and each entry of $\vec{p}$ is strictly positive, this defines an irreducible Markov chain $X$ with state space $V_\Gamma$ with transition probabilities
    \begin{equation}
        \operatorname{Prob}(X = w\,\,\vert\,\, X = v)
        =
        p_{\seq{v, w}},
    \end{equation}
    which follows because each element of $\vec{p}$ is strictly positive and $\Gamma$ is strongly connected \cite[Ch. I.9]{Schaefer1974BanachOperators}. 
    In particular, if we let $\boldsymbol{p}$ be the unique invariant measure of this Markov chain, let $\Omega = V_\Gamma^\mbZ$ and $\mcF$ be the $\sigma$-algebra generated by cylinder sets, and let $T:\Omega\to\Omega$ be the shift $T\seq{\seq{v_n}_{n\in\mbZ}}
    =
    \seq{v_{n+1}}_{n\in\mbZ}$, then letting $\mbP$ denote the probability measure on $\Omega$ induced by $\seq{\boldsymbol{p}, \vec{p}}$, we have that $T:\Omega\to\Omega$ is ergodic for $\mbP$ \cite[Ch. 7.2]{Viana2015FoundationsTheory}. 
    Within this Markovian framework, we can easily construct the example we described above. 
    Indeed, the goal is to use the Markovian structure to dictate that if $\phi_k$ is PPT, then none of $\phi_{k+1}, \dots, \phi_{k+n-1}$ is PPT. 
    To be more precise, let $V_{\mathrm{PPT}}\subsetneq V_\Gamma$ be a distinguished nonempty set of vertices. 
    For all $v\in V_{\mathrm{PPT}}$, let $\phi_v$ be a PPT bcp map, and for all $v\in V_\Gamma\setminus V_{\mathrm{PPT}}$, let $\phi_v$ be a \textit{non}-PPT bcp map. 
    Define $\phi:\Omega\to\bcp$ by 
    \begin{equation}
        \phi_{\seq{v_n}_{n\in\mbZ}}
            :=
        \phi_{v_0}.
    \end{equation}
    Then by Theorem \ref{Thm:PPT}, we may conclude that, for $\mbP$-almost every $\seq{v_n}_{n\in\mbZ}\in\Omega$, we have that 
    \begin{equation}
        \lim_{n\to\infty}
            d_{\operatorname{HS}}(\phi_{v_n}\circ\cdots\circ \phi_{v_0}, \eb)
                =
            0.
    \end{equation}
    This construction gives us great freedom in choice of frequency of occurrence of PPT maps in the above compositions, so that, for example, we may force these compositions to never contain consecutive compositions of PPT maps. 
    It is not hard to see that, given any $k\in\mbN$, we may construct $\Gamma$ in such a way that, given $\phi_{v_0}$ is PPT (i.e., given $v_0\in V_{\mathrm{PPT}}$), we have that $\phi_{v_i}$ is not PPT for all $i = 1, \dots, k-1$ (i.e., the shortest path from $v_0$ back to $V_{\mathrm{PPT}}$ is of length at least $k$). 
\end{example}
\begin{remark*}
    There is a natural question implicit in the above discussion of asymptotics, highlighted by the above example:
    what is the rate of convergence in the asymptotic approach to the set $\eb$?
    In particular, it seems interesting to ask how the quantity $\mu[\phi\text{ is PPT}]$ is related to 
    \begin{equation}\label{Eqn:remark}
        \begin{split}
            f:\mbN\times (0, \infty) &\to [0, 1]\\
            f(n, \varepsilon)
                &:=
            \mu\Big[
            d_{\operatorname{HS}}\seq{
            \Phin, \eb 
            }
            < \varepsilon
            \Big],
        \end{split}
    \end{equation}
    and how the set PPT functions in which $\phi$ takes values affects $f$.
    Indeed, an alternative, equivalent version of the $\mathrm{PPT}^2$ conjecture (see \cite[Conjecture 4.1]{Christandl2019WhenBreaking}) states that the composition of any two PPT maps $\varphi\circ\psi$ is entanglement breaking, and therefore, assuming this conjecture holds, it is obvious that the approach to $\eb$ is immediate given any consecutive compositions of PPT maps. 
    However, in the case where there are no consecutive compositions of PPT maps, as in the above example, the quantity (\ref{Eqn:remark}) becomes the primary quantity of interest.
    This seems to be an interesting route for further exploration. 
\end{remark*}
Next, we change gears to the question of eventual entanglement breaking. 
Recall the definition of eventual entanglement breaking. 
\begin{definition}[Eventual entanglement breaking and index of separability]
    Let $\Psi = \seq{\psi_n}_{n\in\mbZ}$ be a bi-infinite sequence of linear maps $\psi_n:\matrices\to\matrices$. 
    If there exists $N\in\mbN$ such that $\psi_{N-1}\circ\cdots\circ\psi_0\in\eb$, we say that $\Psi$ is eventually entanglement breaking (e.e.b.). 
    We define $\iota(\Psi))$ to be the minimal $N\in\mbN$ such that $\psi_{N-1}\circ\cdots\circ\psi_0\in\eb$, and if no minimum exists we define $\iota(\Psi) = \infty$.
    We call $\iota(\Psi)$ the index of separability of $\Psi$. 
    We call $\Phi$ e.e.b. at $\omega\in\Omega$ if $\Phi_\omega$ is e.e.b., and we say that $\Phi$ is e.e.b. if $\Phi_\omega$ is e.e.b. for almost every $\omega\in\Omega$. 
\end{definition}
By abuse of notation, we write $\iota:\Omega\to\mbN\cup\set{\infty}$ to denote the random number $\iota\!\seq{\Phi}$, and we call $\iota$ the index of separability of $\Phi$. 
It is clear from the definition of $\iota$ that $\iota$ is a $\Phi$-stopping time, and, by a simple application of Theorem \ref{Thm:EB}, we can already prove Theorem \ref{Thm:Index_of_serparbility}. 
\Separability*
\begin{proof}
    As noted above, it is clear that $\iota$ is a $\Phi$-stopping time, so it just remains to show that $\mu[\iota < \infty]\in\set{0, 1}$. 
    To do this, we show that $T^{-1}\set{\iota < \infty}\subset \set{\iota<\infty}$, so that by Proposition \ref{Prop:Basics_of_ergodicity} we will conclude that $\mu[\iota < \infty]\in\set{0, 1}$.
    To see this fact, assume that $\iota(\omega)<\infty$. 
    Then there is $n\in\mbN$ such that $\Phin_\omega\in\eb$. 
    Therefore, by Theorem \ref{Thm:EB} (c), it holds that $\Phi^{(n+1)}_{T^{-1}(\omega)} = \Phin_\omega\circ\phi_{T^{-1}(\omega)}\in\eb$. 
    Therefore, $n+1$ satisfies $\iota(T^{-1}(\omega))\leq n+1 < \infty$. 
    This shows that $T^{-1}\set{\iota < \infty}\subset \set{\iota<\infty}$, and, as described above, the result follows.
\end{proof}
We now set out to prove our final theorem, Theorem \ref{Thm:EEB}. 
Recall that a quantum channel $\psi:\matrices\to\matrices$ is called strictly positive if $\psi(p)$ is full rank whenever $p\geq 0$ is nonzero, that $\psi$ is called primitive if there exists $n\in\mbN$ such that $\psi^n$ is strictly positive, and $\psi$ is called irreducible if there is a unique density matrix $\rho$ such that $\psi(\rho) = \rho$ and $\rho$ is positive definite.
In the nondisordered theory, the well-developed theory of repeated compositions of \textit{fixed} quantum channels (see \cite{MichaelM.Wolf2012QuantumTour}) allows one to conclude that any primitive quantum channel is eventually entanglement breaking. 
A generalization of primitivity for compositions of random quantum channels is the notion of eventual strict positivity mentioned in the introduction: 
\begin{description}
\hypertarget{ESPa}{}
    \item[\ESPa] For almost every $\omega\in\Omega$, there exists $N_0\in\mbN$ such that for all $n\geq N_0$, $\Phin_\omega$ is strictly positive. 
\end{description}
I.E., if $\Phi$ satisfies \ESPa, we call $\Phi$ eventually strictly positive. 
Our main technical result here is the following result that refines \cite[Theorem 2]{Movassagh2022AnStates} in the bistochastic case. 
\begin{prop}\label{Prop:ESP}
    Assume $\phi\in\bcp$ almost surely. 
    Let $\mcA_\Phi$ denote the stabilized multiplicative domain of $\Phi$. 
    Then $\Phi$ satisfies \ESPa\, if and only if $\StabMultDom{\Phi} = \mbC\mbI$ almost surely. 
    Moreover, if these equivalent conditions hold, there is deterministic $\gamma\in (0, 1)$ such that for almost every $\omega\in\Omega$, we have
    \begin{equation}\label{Eqn:Thm:ESP, eqn 1}
        \hsnorm{
        \Phin_\omega
        -
        \Delta_1
        }
        \leq 
        C_{\gamma; \omega} \gamma^n,
    \end{equation}
    for all $n\in\mbN$, where $\Delta_1$ is the map $\Delta_1(a) = d^{-1}\tr{a}\mbI$ and $C_{\gamma}:\Omega\to(0, \infty)$ is measurable.
\end{prop}
\begin{proof}
    We start by showing that \ESPa\, if and only if $\StabMultDom{\Phi} = \mbC\mbI$ almost surely.
    So, assume that $\Phi$ satisfies \ESPa. 
    Then for almost every $\omega\in\Omega$, there is $N\in\mbN$ such that $\Phi^{(N)}_\omega$ is strictly positive. 
    Because $\Phi^{(N)}_\omega$ is a quantum channel, we have that $\Phi^{(N)*}_\omega$ is also strictly positive. 
    Thus, $\Phi^{(N)*}_\omega\circ\Phi^{(N)}_\omega$ is strictly positive.
    Because strictly positive maps are irreducible \cite[Ch. 6.3]{MichaelM.Wolf2012QuantumTour}, we may conclude that $\Phi^{(N)*}_\omega\circ\Phi^{(N)}_\omega$ is irreducible, so by \cite[Lemma 2.1]{Rahaman2017MultiplicativeChannels} we have that $\mcM_{\Phi^{(N)}_\omega} = \mbC \mbI$.
    Because this holds for almost every $\omega\in\Omega$, we conclude that $\StabMultDom{\Phi}\subseteq\mbC\mbI$ almost surely.
    However, $\StabMultDom{\Phi}$ is a unital algebra, hence $\StabMultDom{\Phi} = \mbC\mbI$ almost surely, as claimed. 
    Conversely, assume $\StabMultDom{\Phi} = \mbC\mbI$ almost surely.
    Our strategy to show that $\Phi$ satisfies \ESPa is to apply Theorem \ref{Thm:Main theorem} to conclude that $\Phin_\omega$ is asymptotic to the replacement channel $\Delta_1$ described in the statement of the proposition above, and then to use continuity arguments to conclude that $\Phin_\omega(\rho)$ is positive definite for some $n$ almost surely for any density matrix $\rho\in\matrices$. 
    So, let $\rho\in\matrices$ be a density matrix. 
    Then for any $n\in\mbN$, we have that 
    \begin{equation}\label{Eqn:ESP, eqn 1}
        \Phin_\omega(\rho)
        =
        d^{-1}\mbI 
        +
        \Phin_\omega\seq{\rho - d^{-1}\mbI},
    \end{equation}
    because $\Phin_\omega$ is unital almost surely.
    Note that $\rho - d^{-1}\mbI\in\StabMultDomPerp{\Phi; \omega} = \seq{\mbC\mbI}^\perp$, since $\seq{\mbC\mbI}^\perp$ is the set of matrices with zero trace.
    Now, because the set $\seq{\mbC\mbI}^\perp\cap\set{a\in\matrices\,\,:\,\,\hsnorm{a} \leq \sqrt{2]}}$ is compact, from Theorem \ref{Thm:Main theorem}, if we let $\kappa = \lambda_2$ be the second Lyapunov exponent of $\Phi$, then there is a constant $C>0$ independent of $\rho$ such that 
    \begin{equation}\label{Eqn:ESP, eqn 2}
        \infnorm{\Phin_\omega\seq{\rho - d^{-1}\mbI}}
        \leq 
        C e^{-n\kappa}
    \end{equation}
    for all $n$, uniformly in $\rho$. 
    Therefore, from (\ref{Eqn:ESP, eqn 1}) we can conclude that there is some $N$ for which $ \sigma\seq{\Phi^{(N)}_\omega(\rho)}\subset (0, \infty)$ holds for all $\rho$. 
    Because this holds almost surely, we conclude $\Phi$ satisfies \ESPa. 
    The next part of the theorem follows directly from \cite[Theorem 2]{Movassagh2022AnStates}, but we give a proof here for completeness. 
    First, note that by the same argument that lead to (\ref{Eqn:ESP, eqn 2}), we know that for almost every $\omega\in\Omega$, there is a constant $C_\omega' \in (0, \infty)$ for which (\ref{Eqn:ESP, eqn 2}) holds. 
    Therefore, from (\ref{Eqn:ESP, eqn 1}) and (\ref{Eqn:ESP, eqn 2}), we conclude that 
    \begin{equation}
        \sup_{\rho\in\set{\text{density matrices}}}
        \hsnorm{
        \Phin_\omega(\rho)
        -
        \Delta_1(\rho)
        }
        \leq  
        C_\omega' e^{-n\kappa}
    \end{equation}
    holds for almost every $\omega\in\Omega$. 
    Because any $a\in\matrices$ can be written as a linear combination of four density matrices, $a = c_1 \rho_1 + c_2\rho_2 + c_3\rho_3 + c_4\rho_4$, with each $|c_i|\leq \infnorm{a}$, we may conclude from this equation that 
    \begin{equation}
        \sup_{a\in\matrices\,\,:\,\,\hsnorm{a} = 1}
        \hsnorm{
        \Phin_\omega(a)
        -
        \Delta_1(a)
        }
        \leq 
        4C''
        C_\omega'
        e^{-n\kappa}
    \end{equation}
    holds almost surely, where $C''>0$ is a universal constant giving the equivalence of norms $\hsnorm{\cdot}$ and $\infnorm{\cdot}$. 
    Thus, $C_\omega =  4C''
        C_\omega'$ and $\gamma = e^{-n\kappa}$ fulfills the requirements of the proposition, which concludes the proof. 
\end{proof}
This proposition finally allows us to conclude the version of Theorem \ref{Thm:EEB} without assumption \Indep. 
\begin{thm}
    Assume \BCPa. 
    If $\StabMultDom{\Phi} = \mbC\mbI$ almost surely, then $\Phi$ is e.e.b.
\end{thm}
\begin{proof}
    Let $\Delta_1:\matrices\to\matrices$ be the map $\Delta_1(a) = d^{-1}\tr{a}\mbI$, which is clearly entanglement breaking. 
    By arguing as in \cite[Discussion following Theorem IV.2]{Rahaman2018EventuallyMaps}, from \cite{Gurvits2002}, we know that there is an open neighborhood $W$ of $\Delta_1$ such that $\psi\in W$ implies $\psi\in\eb$.
    Thus, by Proposition \ref{Prop:ESP}, we conclude that almost surely there exists some finite $N\in\mbN$ such that $\Phi^{(N)}\in W$, hence $\Phi^{(N)}\in\eb$. 
    That is, $\Phi$ is eventually entanglement breaking, as claimed. 
\end{proof}
%
%
%
%
%

%% file: sd_Final.tex
We conclude with some final remarks. 
We first discuss how the assumption of bistochasticity is used in the above, and see how, if we drop this assumption, how we might be able to prove a theorem like Theorem \ref{Thm:Main theorem}. 
Then we describe a result of independent interest that generalizes a fact proved by Kuperberg. 
\subsection{The bistochasticity assumption}
In \cite{Ekblad2024ReducibilityInteractions}, it is shown that for any ergodic quantum process $\Phi$ defined by $\seq{T, \phi}$ where $\phi_\omega$ is unital almost surely, there is a random density matrix $\varrho:\Omega\to\matrices$ such that 
\begin{equation}\label{Eqn:Steady_state}
    \phi^*_{T(\omega)}(\varrho_{T(\omega)})
    =
    \varrho_{\omega}
\end{equation}
holds almost surely. 
In general, however, $\varrho$ may not be full rank and may not be deterministic. 
Within this context, therefore, the bistochasticity assumption amounts to taking $\varrho = \mbI$ almost surely. 
It is certainly a natural question to ask whether a Theorem \ref{Thm:Main theorem} can be proved under only the assumption of almost sure unitality, but owed to the fact that $\varrho$ need not even be full-rank almost surely, there would need to be some substantial modifications made to the arguments we gave here in order to prove it. 
This is an interesting route for further work, and, in order to more concretely highlight the differences between the bistochastic regime and the general trace-preserving regime, we now give some basic results towards generalizing Rahaman's result (Theorem \ref{Thm:Rahaman}) \textit{without} the assumption of trace-preservation, thereby giving a rough layout for how the disordered case might go. 
A first step is to notice that by the Perron-Frobenius theory of positive maps on finite-dimensional $C^*$-algebras \cite{Evans1978SpectralC-Algebras}, for any unital completely positive map $\psi:\matrices\to\matrices$, there is a (not necessarily unique) density matrix $\varrho\in\matrices$ such that 
\begin{equation}
    \psi^*(\varrho) = \varrho.
\end{equation}
To refer to such a $\varrho$, we call $\varrho$ a Perron-Frobenius eigenmatrix of $\psi^*$. 
We shall make the assumption that $\varrho$ is full-rank, so that, in particular, it defines a nondegenerate inner product $\inner{\cdot}{\cdot}_\varrho$ on $\matrices$ by the rule 
\begin{equation}
    \inner{a}{b}_\varrho 
    =
    \tr{\varrho b^* a}.
\end{equation}
We write $\|\cdot\|_\varrho$ to denote the norm induced by $\inner{\cdot}{\cdot}_\varrho$, and give a subset $V\subset\matrices$, we write $V^{\perp_\varrho}$ to denote the set $\set{a\in\matrices\,\,:\,\,\inner{v}{a}_\varrho\,\text{for all $v\in V$}}$.
Then we have the following generalization of Lemma \ref{Lem:Mult dom of bcp map}. 
\begin{lemma}\label{Lem:Final_remarks}
    For a unital completely positive map $\psi$ with a full-rank Perron-Frobenius eigenmatrix $\varrho$ of $\psi^*$, we have
    \begin{equation}
        \mcM_\psi
        =
        \set{a\in\matrices\,\,:\,\, \|a\|_\varrho 
        =
        \|\psi(a)\|_\varrho
        \text{ and }
        \|a^*\|_\varrho 
        =
        \|\psi(a^*)\|_\varrho
        }
    \end{equation}
    and
    \begin{equation}
         \mcM_\psi^{\perp_\varrho}\setminus\set{0}
         \subset 
         \set{
         a\in\matrices\,\,:\,\, \|a\|_\varrho > \|\psi(a)\|_\varrho
         \text{ and }
         \|a^*\|_\varrho > \|\psi(a^*)\|_\varrho
         }.
    \end{equation}
    In particular, $\|\psi\|_\varrho \leq 1$. 
\end{lemma}
\begin{proof}
    Assume $a\in\mcM_\psi$. 
    Then from Theorem \ref{Thm:Choi} (which applies to any unital, not necessarily trace-preserving maps), we have that $\psi(a^*a) = \psi(a)^*\psi(a)$ and $\psi(aa^*) = \psi(a)\psi(a)^*$. 
    Thus, since $\psi^*(\varrho) = \varrho$, we conclude that $\tr{\varrho\psi(a^*a)}
        =
        \tr{\varrho a^*a}$, i.e., $\|a\|_\varrho = \|\psi(a)\|_\varrho$.
        Similarly, because $\psi(aa^*) = \psi(a)\psi(a)^*$, we conclude $\|a^*\|_\varrho = \|\psi(a^*)\|_\varrho$.
    Conversely, suppose $\|a\|_\varrho = \|\psi(a)\|_\varrho$ and $\|a^*\|_\varrho 
        =
        \|\psi(a^*)\|_\varrho$.
    By the Schwarz inequality and the fact that $\varrho\geq 0$, we know that $\varrho^{1/2}\psi(a^*a)\varrho^{1/2} - \varrho^{1/2}\psi(a)^*\psi(a)\varrho^{1/2}\geq 0$, so because $\psi^*(\varrho) = \varrho$, we conclude that
    \begin{equation}
        \tr{
        \varrho^{1/2}\psi(a^*a)\varrho^{1/2} - \varrho^{1/2}\psi(a)^*\psi(a)\varrho^{1/2}
        }
        =
        \|a\|_\varrho^2
        -
        \|\psi(a)\|_\varrho^2
        =
        0.
    \end{equation}
    Thus, because $\varrho$ is full rank, this implies that $\psi(a^*a) = \psi(a)^*\psi(a)$. 
    By arguing the same for $a^*$ using that $\|a^*\|_\varrho 
        =
        \|\psi(a^*)\|_\varrho$, we conclude $\psi(aa^*) = \psi(a)\psi(a)^*$.
    To see the last part, notice that the Schwarz inequality gives 
    \begin{equation}
        \|a\|_\varrho^2 
        =
        \tr{\varrho a^*a}
        =
        \tr{\varrho \psi(a^*a)}
        \geq 
        \tr{\varrho \psi(a)^*\psi(a)}
        =
        \|\psi(a)\|_\varrho^2
    \end{equation}
    for any $a\in\matrices$. 
    So, since $\mcM_\psi\cap\mcM_{\psi}^{\perp_\varrho} = \set{0}$, by the above any $a\in \mcM_{\psi}^{\perp_\varrho}$ must satisfy $\|a\|_\varrho > \|\psi(a)\|_\varrho$ and $\|a^*\|_\varrho > \|\psi(a^*)\|_\varrho$, which concludes the proof. 
\end{proof}
A corollary of this is the following analog of Lemma \ref{Lem:Comp of mult doms}.
\begin{lemma}\label{Lem:Final_remakrs_2}
    For a unital completely positive map $\psi$ such that there is a full-rank Perron-Frobenius eigenmatrix $\varrho$ for $\psi^*$, $\mcM_{\psi^{n+1}} = \set{a\in\mcM_{\psi^n}\,\,:\,\,\psi(a)\in\mcM_{\psi}}$ for all $n\in\mbN$. 
\end{lemma}
\begin{proof}
    Let $\varrho$ be a fixed Perron-Frobenius eigenmatrix of $\psi^*$. 
    Then $\varrho$ is also a Perron-Frobenius eigenmatrix of $(\psi^*)^n$ for all $n$. 
    In particular, by the Schwarz inequality we have that 
    \begin{equation}\label{Eqn:Final_remarks}
        \tr{\varrho \psi^{n+1}(a)^*\psi^{n+1}(a)}
        \leq 
        \tr{\varrho \psi^{n}(a)^*\psi^{n}(a)}
        \leq 
        \tr{\varrho a^*a}
    \end{equation}
    for any $a\in\matrices$. 
    That is, $\| \psi^{n+1}(a)\|_\varrho\leq \| \psi^{n}(a)\|_\varrho\leq \| a\|_\varrho$. 
    So, by Lemma \ref{Lem:Final_remarks}, because $a\in\mcM_{\psi^{n+1}}$ if and only if $\| \psi^{n+1}(a)\|_\varrho = \| a\|_\varrho$, the inequalities in (\ref{Eqn:Final_remarks}) all become equalities, and another application of Lemma \ref{Lem:Final_remarks} concludes the proof.  
\end{proof}
With these two lemmas in hand, one might now attempt to emulate the proof we gave for Theorem \ref{Thm:Main theorem} in order to provide a suitable generalization of Theorem \ref{Thm:Rahaman}, replacing in each place the usage of Lemma \ref{Lem:Mult dom of bcp map} by Lemma \ref{Lem:Final_remarks} and Lemma \ref{Lem:Comp of mult doms} by Lemma \ref{Lem:Final_remakrs_2}, and using the random steady state $\varrho$ from (\ref{Eqn:Steady_state}).
As we have already seen, however, we have required the assumption that $\varrho$ in the above be full-rank. 
So, there are at least two technical obstructions here: 
first, one must determine how to replace the usage of $\innerhs{\cdot}{\cdot}$ with $\inner{\cdot}{\cdot}_\varrho$, and second, one must determine how to deal with the case that $\varrho$ is not even necessarily full-rank. 
It is probable these problems can be overcome, although it will likely take new arguments to do so. 
\subsection{A generalization of a theorem of Kuperberg}
Changing directions to an entirely different issue, we take this space to give a proposition that seems to be of independent interest that generalizes the result \cite[Theorem 2.1]{Kuperberg2003TheMemory} of Kuperberg, at least in the bcp case.
We say $\psi:\matrices\to\matrices$ is the conditional expectation onto a subalgebra $\mcA\subseteq\matrices$ if $\psi$ is the orthogonal projection (in the Hilbert-Schmidt sense) onto $\mcA$.
We then are able to prove the following result.
\begin{prop}
    Assume $\phi\in\bcp$ almost surely, and that $\seq{\phi_n}_{n\in\mbZ}$ is an i.i.d. sequence. 
    Let $\mcA_\phi$ denote the stabilized multiplicative domain of $\Phi$. 
    If we let $\psi:\matrices\to\matrices$ denote the conditional expectation onto $\StabMultDom{\Phi}$, then almost surely $\seq{\phi_n\circ\cdots\circ\phi_1}_{n\in\mbN}$ contains a subsequence converging to $\psi$. 
\end{prop}
By Theorem \ref{Thm:Main theorem}, to prove this theorem it suffices to note the following fact from the theory of random walks on compact groups, which generalizes the pedestrian fact that, given $\lambda\in\mbT$, there is a sequence $\seq{n_k}_{k\in\mbN}$ such that $\lambda^{n_k}\to 1$.
For a Hilbert space $\scrK$, let $\mbU\seq{\scrK}$ denote the collection of unitary operators $U:\scrK\to\scrK$. 
\begin{prop}\label{Prop:Products of random unitaries}
     Let $\seq{\Xi, \mcG, \nu}$ be a standard probability space, and let $\scrK$ be a finite-dimensional Hilbert space. 
     Suppose that $\seq{U_n}_{n\in\mbN}$ is an i.i.d. sequence of unitary matrices $U_n:\Xi\to\mbU\seq{\scrK}$. 
     Then almost surely $\seq{U_n\cdots U_1}_{n\in\mbN}$ contains a subsequence converging to the identity map $\operatorname{Id}_\scrK$ of $\scrK$.  
\end{prop}
\begin{proof}
    Let $\nu$ denote the law of $U_1$, which is a Borel probability measure on the compact second countable group $\mbU\seq{\scrK}$, and let $G$ denote the closure of the group generated by the support of $\nu$, which is itself a compact second countable group.
    In \cite{Guivarch2012RecurrenceSpaces} (specifically, the discussion leading up to \cite[Theorem 1.1]{Guivarch2012RecurrenceSpaces}, where we note that $G$ is $\nu$-adapted and locally compact with polynomial growth of degree $d = 0$), it was shown that for all $g\in G$ the random walk on $G$ defined by 
    \begin{equation}
        \begin{split}
            X_n: \Xi &\to\mbU\seq{\scrK}\\
            X_n\seq{\xi} &:= U_n(\xi)\cdots U_1(\xi)
        \end{split}
    \end{equation}
    is \textit{recurrent}, i.e., for all neighborhoods $W\subseteq G$ containing the identity $\operatorname{Id}_\scrK\in G$ and for $\nu$-almost every $\xi\in\Xi$, we have that 
    \begin{equation}
        \#\set{n\in\mbN\,\,:\,\, X_n(\xi)\in W}
        =
        \infty. 
    \end{equation}
    Because this is almost surely true for arbitrary neighborhoods $W$ of $\operatorname{Id}_\scrK$, it is clear that almost surely $\seq{U_n\cdots U_1}_{n\in\mbN}$ contains a subsequence converging to $\operatorname{Id}_\scrK$, which concludes the proof. 
\end{proof}

%% file: sdapp_Grass.tex
Recall that, given a finite-dimensional $\mbC$-vector space $V$ with a nondegenerate inner product $\inner{\cdot}{\cdot}_V$, the Grassmannian $\mbG(V)$ of $V$ is the collection of subspaces of $V$.
It is well-known that $\mbG(V)$ has the structure of a metric space when equipped with the Grassmannian metric 
\begin{equation}\label{Eqn:App:Grassmannian metric}
    d_{\mbG(V)}\seq{W, U}
        =
    \max 
    \Bigg( 
        \sup_{w\in W_1}\inf_{u\in U} \|w - u\|_V, 
        \sup_{u\in U_1}\inf_{w\in W}\|w - u\|_V
    \Bigg)
\end{equation}
for $W, U\in\mbG(V)$, where $\|\cdot\|_V$ denotes the norm induced by $\inner{\cdot}{\cdot}_V$ and $S_1 = \{s\in S\,\,:\,\, \|S\|_V = 1\}$ for a subset $S\subset V$.
We call the Borel $\sigma$-algebra on $\mbG(V)$ that this metric space structure induces the Grassmannian $\sigma$-algebra (for $V$).
Note that $\matrices$ is a finite-dimensional vector space with the Hilbert-Schmidt inner product $\inner{\cdot}{\cdot}_{\operatorname{HS}}$, so this construction applies to $\mbG\seq{\matrices}$. 
In this section, we establish the measurability of the multiplicative domain and the stabilized multiplicative domain as mappings $\bcp\to\mbG(\matrices)$. 
To do this, it is convenient to discuss another measurability structure for some collection of subsets of $\matrices$. 
Recall the definition of the Fell $\sigma$-algebra \cite[Ch. 1.1.1]{Molchanov2017TheorySets}. 
\begin{definition}[Fell $\sigma$-algebra]
    Let $\mbE$ be a locally compact Hausdorff second countable space, and let $\scrF(\mbE)$ denote the set of closed subsets of $\mbE$. 
    Define the Fell $\sigma$-algebra $\mcB_{\scrF\seq{\mbE}}$ on $\scrF(\mbE)$ to be the topology generated by sets of the form 
    \begin{equation}
        \mfF_K  
            :=
        \set{F\in\scrF\seq{\mbE}\,\,:\,\,
        F\cap K\neq\emptyset},
    \end{equation}
    where $K$ varies over all compact subsets of $\mbE$. 
\end{definition}
By local compactness, we note that, for any open set $G\subseteq\mbE$ and compact set $K\subseteq \mbE$, sets of the form
\begin{equation}
    \mfF_G
    =
    \set{F\in\scrF(\mbE)\,\,:\,\, F\cap G\neq\emptyset}
    \quad\text{and}\quad
    \mfF_{G\cap K}
    =
    \set{F\in\scrF(\mbE)\,\,:\,\, F\cap G\cap K\neq\emptyset}
\end{equation}
are Fell-measurable (see \cite[Ch. 1.1.1]{Molchanov2017TheorySets}), a fact which we shall use later. 
We then have the following. 
\begin{lemma}\label{Lem:App:Mult dom is Fell meas}
    The map $\mcM:\bcp\to \scrF\seq{\matrices}$ defined by $\psi\mapsto \mcM_\psi$ is measurable with respect to the Fell $\sigma$-algebra.  
\end{lemma}
\begin{proof}
    Let $K\subset\matrices$ be compact. 
    We need to show that 
    \begin{equation}
        \mcM^{-1}\seq{\mfF_K}
            =
        \set{\psi\in\bcp\,\,:\,\,
        \mcM_\psi\cap K\neq\emptyset}
    \end{equation}
    is measurable with respect to the Borel $\sigma$-algebra on $\matrices$. 
    In fact, we have that $\mcM^{-1}\seq{\mfF_K}$ is a closed set (hence measurable): indeed, suppose that $\seq{\psi_n}_{n\in\mbN}$ is a convergent sequence in $\mcM^{-1}\seq{\mfF_K}$ with limit point $\psi\in\bcp$. 
    Then for all $n\in\mbN$, there is $a_n\in\mcM_{\psi_n}\cap K$. 
    By the compactness of $K$, therefore, there is a convergent subsequence $\seq{a_{n_k}}_{k\in\mbN}$ of $\seq{a_n}_{n\in\mbN}$, converging to some $a\in K$. 
    Thus, by Lemma \ref{Lem:Mult dom of bcp map}, we have that 
    \begin{equation}
        \|\psi(a)\|_{\operatorname{HS}}
        =
        \lim_{k\to\infty}\|\psi_{n_k}(a_{n_k})\|_{\operatorname{HS}}
        =
        \lim_{k\to\infty}\|a_{n_k}\|_{\operatorname{HS}}
        =
        \|a\|_{\operatorname{HS}},
    \end{equation}
    so by another application of Lemma \ref{Lem:Mult dom of bcp map} we conclude that $a\in\mcM_\psi$. 
    That is, $\mcM_\psi\cap K\neq\emptyset$, hence $\psi\in \mcM^{-1}\seq{\mfF_K}$, which shows that $\mcM^{-1}\seq{\mfF_K}$ is closed, concluding the proof. 
\end{proof}
With these preliminary results in hand, we are positioned to give a proof of the measurability of all the maps associated to $\mcM$ we consider in this paper. 
Note the trivial fact that all subspaces are closed sets, hence $\mbG\seq{\matrices}\subset\scrF\seq{\matrices}$, so in particular the Fell $\sigma$-algebra induces a $\sigma$-algebra on $\mbG\seq{\matrices}$.
\begin{lemma}\label{Lem:App:Measurability of mult domains}
    Let $\seq{\Xi, \mcA}$ be a measurable space and suppose that $\mcV: \Xi\to \mbG\seq{\matrices}$ defines a subspace-valued random variable that is measurable with respect to the Fell $\sigma$-algebra. 
    Then $\mcV$ defines a subspace-valued random variable that is measurable with respect to the Grassmannian metric. 
\end{lemma}
\begin{proof}
    To see that the subspace-valued function $\mcV:\Xi\to\mbG\seq{\matrices}$ is measurable with respect to the Grassmannian metric, it suffices to show that 
    \begin{equation}
    \begin{split}
        \Xi &\to[0, \infty)\\
        \xi&\mapsto d_{\mbG\seq{\matrices}}\seq{\mcV_\xi, W}
    \end{split}
    \end{equation}
    defines a measurable function for any $W\in\mbG\seq{\matrices}$, as the collection of balls generates the Grassmannian $\sigma$-algebra. 
    So, fix  $W\in\mbG\seq{\matrices}$. 
    For $i = 1, 2$, define $d_i:\Xi\to [0, \infty)$ by 
    \begin{equation}
    d_{i}(\xi)
        :=
        \begin{cases}
            \phantom{..}\displaystyle\sup_{w\in W_1}\inf_{a\in\mcV_\xi}\|w - a\|_{\operatorname{HS}} &\text{if $i = 1$}\\
            \displaystyle\sup_{a\in\seq{\mcV_\xi}_1}\inf_{w\in W}\|w - a\|_{\operatorname{HS}} &\text{if $i = 2$}.
        \end{cases}
    \end{equation}
    Since $d_{\mbG\seq{\matrices}}\seq{\mcV_\xi, W} = \operatorname{max}\seq{d_1(\psi), d_2(\psi)}$, it suffices to show that $d_1$ and $d_2$ are measurable functions of $\xi$. 
    First, some notation:
    for a nonempty subset $S\subseteq\matrices$, let $d_{\operatorname{HS}}(\cdot, S)$ denote the function defined by 
    \begin{equation}
        \begin{split}
            d_{\operatorname{HS}}(\cdot, S):\matrices &\to [0, \infty)\\
            d_{\operatorname{HS}}(a, S) &= \inf_{s\in S}\hsnorm{a- s}.
        \end{split}
    \end{equation}
    In the particular case that $S$ is a closed set, it is easy to see that $d_{\operatorname{HS}}(\cdot, S)$ is a continuous function.  
    We shall use this fact in the following. 
    To see that $d_1$ is measurable, fix a countable dense subset $\set{w_n}_{n\in\mbN}$ of $W_1$. 
    Because $\mcV_\xi$ is a closed subset of $\matrices$, the map $d_{\operatorname{HS}}\seq{\cdot, \mcV_\xi}$ is continuous. 
    In particular, we see that 
    \begin{align*}
        d_1(\xi)
            &= 
        \sup_{w\in W_1}
        d_{\operatorname{HS}}\seq{w, \mcV_\xi}\\
            &= 
        \sup_{n\in\mbN}\,
        d_{\operatorname{HS}}\seq{w_n, \mcV_\xi}.
    \end{align*}
    Thus, we just need to show that $\xi\mapsto d_{\operatorname{HS}}\seq{w_n, \mcV_\xi}$ is measurable for any $n$. 
    To do this, fix $r>0$ and compute
    \begin{align*}
        \set{\xi\in\Xi\,\,:\,\,\inf_{a\in\mcV_\xi}\hsnorm{w_n - a} \geq r}
            &= 
         \set{\xi\in\Xi\,\,:\,\,\text{for all $a\in\mcV_\xi$, }\hsnorm{w_n - a} \geq r}\\
            &\hspace{-5mm}= 
        \matrices\setminus\set{\xi\in\Xi\,\,:\,\,
            \text{there is $a\in\mcV_\xi$ with }\hsnorm{w_n - a} < r}\\
            &\hspace{-5mm}= 
         \matrices\setminus\set{\xi\in\Xi\,\,:\,\,
         \mcV_\xi\cap B_r(w_n)\neq\emptyset},
    \end{align*}
    which is measurable since $\xi\mapsto\mcV_\xi$ is Fell measurable. 
    Thus, because $r>0$ was arbitrary, we may conclude that $\xi\mapsto d_1(\xi)$ is measurable. 
    To see that $d_2$ is measurable, again fix $r > 0$ and compute 
    \begin{align*}
        \set{\xi\in\Xi\,\,:\,\,
            d_2(\xi) > r}
                &= 
        \set{\xi\in\Xi\,\,:\,\,
            \text{there is $a\in\seq{\mcV_\xi}_1$ with }d_{\operatorname{HS}}(a, W) > r}
                    \\
                &= 
        \set{\xi\in\Xi\,\,:\,\,
                \mcV_\xi\cap\seq{\matrices}_1\cap d_{\operatorname{HS}}(\cdot, W)^{-1}\big((r, \infty)\big)
                \neq\emptyset
                }.
    \end{align*}
    Thus, since $W$ is a closed set, $d_{\operatorname{HS}}(\cdot, W)^{-1}\big((r, \infty)\big)$ is an open set, so in particular $d_{\operatorname{HS}}(\cdot, W)^{-1}\big((r, \infty)\big)$ may be written as a countable increasing union of compact subsets of $\matrices$.
    Thus, by the Fell measurability of $\xi\mapsto\mcV_\xi$, we may conclude that $\set{\xi\in\Xi\,\,:\,\,
            d_2(\xi) > r}$ is measurable. 
    In particular, because $r>0$ was arbitrary, we conclude that $d_2$ is measurable, and the proof is concluded. 
\end{proof}
Now, note that if $\mcV_n:\Omega\to \scrF\seq{\matrices}$ are Fell-measurable for all $n\in\mbN$, then 
\begin{equation}
    \begin{split}
        \bigcap_{n\in \mbN} \mcV_n:\Omega&\to\scrF\seq{\matrices}\\
        \omega &\mapsto \bigcap_{n\in \mbN} \mcV_{n; \omega}
    \end{split}
\end{equation}
is also Fell-measurable. 
Thus, the above lemma gives us the following. 
\begin{cor}\label{Cor:App:All Ms are Grassmannian measurable}
    The following maps are measurable with respect to the Grassmannian metric. 
    \begin{enumerate}[label = (\alph*)]
        \item The map $\mcM:\bcp\to\mbG\seq{\matrices}$ given by $\psi\mapsto\mcM_\psi$.
        \item The map $\StabMultDom{\phantom{\cdot}}:\bcp\to\mbG\seq{\matrices}$ given by $\psi\mapsto\StabMultDom{\psi}$.
        \item The map $\Omega\to\mbG\seq{\matrices}$ defined by $\omega\mapsto\mcM_{\Phin_\omega}$, for all $n\in\mbN$.
        \item The map $\Omega\to\mbG\seq{\matrices}$ defined by $\omega\mapsto\StabMultDom{\Phi; \omega}$. 
    \end{enumerate}
\end{cor}
Thus, our measurability concerns are alleviated.
We now change gears somewhat, and prove some technical lemmas about the Grassmannian required in the proof of Theorem \ref{Thm:Deterministic}.
For $r\in\set{1, \dots, d^2}$, let $\mbG_r\seq{\matrices}$ be the subset of $\mbG\seq{\matrices}$ consisting of those subspaces of dimension $r$. 
Then it is straightforward to check that, given $V, W\in\mbG_r\seq{\matrices}$, we have that 
\begin{equation}\label{Eqn:App:Simplifcation of Grassmannian metric for rank r}
    d_{\mbG\seq{\matrices}}
    \seq{
    V, W
    }
    =
    \sup_{v\in V_1}\inf_{w\in W}\hsnorm{v - w}.
\end{equation}
Again for $r\in\set{1, \dots, d^2}$, let $\scrP_r$ denote the set
\begin{equation}
    \scrP_r
    =
    \set{
    P\in\scrL\,\,:\,\, P = P^* = P^2
    },
\end{equation}
where we recall $\scrL$ is the set of linear maps $\matrices\to\matrices$, and the adjoint $P^*$ is with respect to $\innerhs{\cdot}{\cdot}$.
Then letting $\infnorm{P}$ denote the operator norm of $P$ with respect to $\hsnorm{\cdot}$, we define a metric $d_\scrP$ on $\scrP_r$ given by 
\begin{equation}
    d_\scrP(P, Q)
    :=
    \infnorm{P - Q}.
\end{equation}
Now, recall the standard bijection between $\scrP_r$ and $\mbG_r\seq{\matrices}$, which is given by the map 
\begin{equation}
    \begin{split}
        \operatorname{Proj}:\mbG_r\seq{\matrices} &\to \scrP_r,
    \end{split}
\end{equation}
where for $W\in\mbG_r\seq{\matrices}$, $\operatorname{Proj}(W)$ is the orthogonal projection onto $W$, i.e., for $a\in\matrices$, $\operatorname{Proj}(W)(a) = b\in W$ is the unique element of $W$ such that $ \hsnorm{a - b}
    =
    \inf_{c\in W}\hsnorm{a - c}.$
It is a standard fact that this is a bijection, and, moreover, this map is an isometry between $\seq{\mbG_r\seq{\matrices}, d_\mbG}$ and $\seq{\scrP_r, d_\scrP}$, a fact we state as a lemma.
\begin{lemma}\label{Lem:App:Isometry_Grassmannian}
    The map $\operatorname{Proj}:\mbG_r\seq{\matrices}\to\scrP_r$ is an isometry between metric spaces $\seq{\mbG_r\seq{\matrices}, d_{\mbG\seq{\matrices}}}$ and $\seq{\scrP_r, d_\scrP}$.
\end{lemma}
\begin{proof}
    See \cite[Ch. III, Section 34]{Akhiezer1993TheorySpace}.
\end{proof}
In the proof of Theorem \ref{Thm:Deterministic}, we also use the following fact. 
\begin{lemma}\label{Lem:App:Lipschitz for Grassmann}
    Let $L:\matrices\to\matrices$ be a linear map. 
    Let $V\in\mbG_r\seq{\matrices}$ and suppose that $L\vert_V$ is isometric with respect to $\hsnorm{\cdot}$.
    Then for any $W\in\mbG_r\seq{\matrices}$, we have that 
    \begin{equation}
        d_{\mbG\seq{\matrices}}(L(V), L(W))
        \leq 
        \hsnorm{L}
        d_{\mbG\seq{\matrices}}(V, W).
    \end{equation}
\end{lemma}
\begin{proof}
    Fix $W\in\mbG\seq{\matrices}$.
    For $v'\in L(V)$ and $w'\in L(W)$, there are $v\in V$ and $w\in W$ so that $v' = L(v)$ and $w' = L(w)$, hence $\hsnorm{v' - w'}\leq\hsnorm{L}\hsnorm{v - w}.$
    In particular, for fixed $v' = L(v)\in \seq{L(V)}_1$, we have that 
    \begin{equation}
        \inf_{w'\in L(W)}
        \hsnorm{v' - w'}
        =
         \inf_{w\in W}
        \hsnorm{L(v) - L(w)}
            \leq 
        \hsnorm{L}
        \inf_{w\in W}
        \hsnorm{v - w}.
    \end{equation}
    On the other hand, because $L$ acts isometrically on $V$, we know that $L\seq{V_1} = \seq{L(V)}_1$.    
    So, from (\ref{Eqn:App:Simplifcation of Grassmannian metric for rank r}) and the above, we conclude that    
    \begin{align*}
        d_{\mbG\seq{\matrices}}
        \seq{L(V), L(W)}
            &= 
        \sup_{v'\in \seq{L(V)}_1}
        \inf_{w'\in L(W)}
        \hsnorm{v' - w'}\\
            &= 
        \sup_{v\in V_1}
        \inf_{w\in W}
        \hsnorm{L(v) - L(w)}\\
        &\leq 
        \hsnorm{L}
        \sup_{v\in V_1}\inf_{w\in W}
        \hsnorm{v - w}\\
        &=
        \hsnorm{L}d_{\mbG\seq{\matrices}}\seq{V, W}.
    \end{align*}
\end{proof}

%% file: sdapp_Stopping.tex
We shall now recall some basic notions from the theory of stochastic processes \cite{Zhan2019LectureProcesses}. 
\begin{definition}[Stochastic processes and stopping times]\label{Def:App:Stoch proc and stopping times}
    Let $\seq{\Xi, \mcA}$ be a measurable space and let $I\subset\mbR$ be an indexing set. 
    A filtration on $I$ is a collection $\mcG = \seq{\mcG_i}_{i\in I}$ of sub-$\sigma$-algebras of $\mcA$ such that $i < j$ implies $\mcG_i\subseteq \mcG_j$. 
    Given another measurable space $\seq{S, \mcS}$, we say that $\seq{X_i}_{i\in I}$ is an $S$-valued stochastic process if $X_i:\seq{\Xi, \mcA}\to \seq{S, \mcS}$ is measurable for all $i$, and we say that $\seq{X_i}_{i\in I}$ is $\mcG$-adapted if $X_i$ is $\mcG_i$-measurable for all $i$. 
    Given a filtration $\mcG$ and a measurable mapping $\iota:\seq{\Xi, \mcA}\to I\cup\{\sup I\}$ (where $I$ is given the $\sigma$-algebra induced by the Borel $\sigma$-algebra on $\mbR$), we say that $\iota$ is an $\mcG$-stopping time if for any $i\in I$, 
    \begin{equation}
        \set{\iota\leq i} = \set{\xi\in\Xi\,\,:\,\, \iota_\xi \leq i}\in \mcG_i. 
    \end{equation}
    Given an $\mcG$-stopping time $\iota$, we let $\mcG_\iota$ denote the $\sigma$-algebra defined by 
    \begin{equation}
        \mcG_\iota 
            := 
        \set{
        A\in\mcA
            \,\,:\,\,
        A\cap\set{\iota\leq i}\in\mcG_i \text{ for all $i\in I$}
        }.
    \end{equation}
    Note that $\mcG_\iota\subset\mcA$ by definition.
\end{definition}
The technical fact from this theory that we shall use is the following. 
\begin{lemma}[Measurability of stopped stochastic process]\label{Lem:App:Stopped process}
    Let $\seq{\Xi, \mcA}$ and $\seq{S, \mcS}$ be measurable spaces, and let $\mcG = \seq{\mcG_i}_{i\in\mbN}$ be filtration on $\mcA$ with index $\mbN$. 
    Let $X = \seq{X_i}_{i\in\mbN}$ be a $\mcG$-adapted stochastic process taking values in $S$. 
    If $\iota:\Xi\to\mbN\cup\{\infty\}$ is a $\mcG$-stopping time, then the map defined by 
    \begin{equation}
        \Xi\ni\xi\mapsto X_{\iota(\xi)}(\xi)
    \end{equation}
    is $\mcG_\iota$-measurable. 
\end{lemma}
\begin{proof}
The proof may be found in \cite[Lemma 6.5]{Zhan2019LectureProcesses}.
\end{proof}
Now consider the stochastic process $\seq{\Phin}_{n\in\mbN}$ taking values in $\bcp$, which is $\mcT$-adapted, where $\mcT = \seq{\mcT_n}_{n\in\mbN}$ is the natural filtration on $\mbN$ defined by $\seq{\Phin}_{n\in\mbN}$, i.e., 
\begin{equation}
    \begin{split}
        \mcT_n 
        =
        \sigma\seq{\Phi^{(m)}\,\,:\,\,m = 1, \dots, n}.
    \end{split}
\end{equation}
Then we have what we need to show that $\omega\mapsto\Phi^{(\tau_\omega)}_\omega$ is measurable. 
\begin{lemma}[Measurability of $\tau$]\label{Lem:App:Meas of tau}
    The map $\omega\mapsto \tau_\omega$ defines a $\mcT$-stopping time. 
    In particular, $\omega\mapsto \Phi^{(\tau_\omega)}_\omega$ is $\mcT_\tau$-measurable (hence $\mcF$-measurable). 
\end{lemma}
\begin{proof}
    By Lemma \ref{Lem:App:Stopped process} it suffices to show that $\omega\mapsto \tau_\omega$ defines a $\mcT$-stopping time.
    To do this, it suffices to show that, for any $m\in\mbN$, we have that 
    \begin{equation}
        \set{\tau = m}\in\mcT_m.
    \end{equation}
    We note that, on the full probability event $D = \set{\delta = \dim\StabMultDom{\Phi}}$, $\tau_\omega = m$ if and only if $\dim \mcM_{\Phi^{(m)}_\omega} = \delta$, where $\delta$ is the constant from Corollary \ref{Cor:Orbit dimension constant on orbits}.
    Because $\seq{\Omega, \mcF,\mu}$ is a complete probability space, therefore, it suffices to show that 
     \begin{equation}
       \set{\omega\in\Omega \,\,:\,\, \dim \mcM_{\Phi^{(m)}_\omega} = \delta}
       \in\mcT_m
    \end{equation}
    for all $m$. 
    This, in turn, is immediate from Lemma \ref{Lem:App:Measurability of mult domains}, since the map 
    \begin{equation}
    \begin{split}
    \Omega &\to \mbG(\matrices)\\
        \omega &\mapsto \mcM_{\Phi^{(m)}_\omega}
        =
        \mcM_{\phi_{T^{m-1}(\omega)\circ\cdots\circ\phi_{\omega}}}
    \end{split}
    \end{equation}
    is expressly $\seq{\mcT_m, \mcB\seq{\scrG}}$-measurable, and the dimension map $\mbG(\matrices)\to\mbN$ is measurable. 
\end{proof}